\documentclass[sigconf]{acmart}

\renewcommand\footnotetextcopyrightpermission[1]{} % removes footnote with conference information in first column

\usepackage[utf8]{inputenc}
\usepackage{xcolor}
\usepackage{xspace}
\usepackage{amsmath}
\usepackage{mathpartir}
\usepackage{microtype}
\usepackage{listings}
\usepackage{stfloats}
\usepackage{enumitem}
\usepackage{balance}

\usepackage{breakurl}

\definecolor{pblue}{rgb}{0.13,0.13,1}
\definecolor{pgreen}{rgb}{0,0.5,0}
\definecolor{pgrey}{rgb}{0.46,0.45,0.48}
\definecolor{pwhite}{rgb}{1.0,1.0,1.0}
\lstdefinelanguage{myML}{
  morekeywords={match,with,proc,assert,for,foreach,range,do,while,until,break,if,then,else,return,and,or,raise,false,true},
  columns=fullflexible,
  sensitive=true,
  commentstyle = \itshape, % \color{blue},
  morecomment={[l]//},
  mathescape=true,
  basicstyle=\small,
  identifierstyle={\ttfamily},
  stringstyle=\rmfamily,
  keywordstyle=\color{black},
  keywordstyle=\bfseries,
  literate={<-}{$\leftarrow\ $}{2} {->}{$\rightarrow\ $}{2} {:=}{$\leftarrow\ $}{2}
}

\newcommand{\ls}{\lstinline[language=myML,basicstyle=\normalsize]}

\newcommand{\eg}{\emph{e.g.}}
\newcommand{\ie}{\emph{i.e.}}

\newcommand{\Var}{\mathit{Var}}
\newcommand{\Fld}{\mathit{Fld}}

\newcommand{\Cls}{\mathit{Cls}}
\newcommand{\Mtd}{\mathit{Mtd}}
\newcommand{\Expr}{\mathit{Expr}}
\newcommand{\Pos}{\mathit{Pos}}
\newcommand{\Reg}{\mathit{Reg}}

\newcommand{\Loc}{\mathit{Loc}}
\newcommand{\Val}{\mathit{Val}}
\newcommand{\Obj}{\mathit{Obj}}
\newcommand{\classOf}{\mathrm{classOf}}

\newcommand{\Object}{\texttt{Object}\xspace}
\newcommand{\NullType}{\texttt{NullType}\xspace}

\newcommand{\itnull}{\mathit{null}}
\newcommand{\this}{\texttt{this}}
\newcommand{\letin}[3]{\texttt{let} \; #1=#2 \; \texttt{in} \; #3}
\newcommand{\ttifthenelse}[4]{\texttt{if} \; #1=#2 \; \texttt{then} \; #3 \; \texttt{else} \; #4}
\newcommand{\ttifthenelseg}[3]{\texttt{if} \; #1 \; \texttt{then} \; #2 \; \texttt{else} \; #3}
\newcommand{\ttnull}{\texttt{null}}
\newcommand{\new}[2]{\texttt{new}^{#1}\; #2}
\newcommand{\cast}[2]{(#1)\, #2}
\newcommand{\tto}{\texttt{emit}}

\newcommand{\fields}{\mathit{fields}}
\newcommand{\methods}{\mathit{methods}}
\newcommand{\mtable}{\mathit{mtable}}

\newcommand{\cP}{\mathcal{P}}
\newcommand{\cPf}{\cP^{\text{fin}}}
\newcommand{\andrel}{\,{\&}\,}
\newcommand{\dom}{\mathrm{dom}}

\newcommand{\N}{\mathbb{N}}
\newcommand{\fA}{\mathfrak{A}}
\newcommand{\vdashb}{\vdash_{\fA}}
\newcommand{\cM}{\mathcal{M}}
\newcommand{\cU}{\mathcal{U}}
\newcommand{\cV}{\mathcal{V}}

\newcommand{\cA}{\mathcal{A}}
\newcommand{\cB}{\mathcal{B}}

\newcommand{\cT}{\mathcal{T}}
\newcommand{\cS}{\mathcal{S}}

\newcommand{\Ma}{M_{\fA}}
\newcommand{\lo}{{\leq\omega}}

\newcommand{\lfp}{\mathsf{lfp}}
\newcommand{\gfp}{\mathsf{gfp}}

\newcommand{\Node}{\texttt{Node}\xspace}
\newcommand{\ttlast}{\texttt{last}\xspace}
\newcommand{\Test}{\texttt{Test}\xspace}
\newcommand{\cyclic}{\texttt{cyclic}\xspace}
\newcommand{\linear}{\texttt{linear}\xspace}

\newcommand{\typeff}{\ls{typeff}}%{\mathsf{typeff}}

\newcommand{\trycatch}[3]{\mathtt{try} \; #1 \; \mathtt{catch}(#2) \; #3}
\newcommand{\throw}[1]{\mathtt{throw} \; #1}

\newcommand{\throws}{\,\mathrel{\mathsf{throws}}\,}

\newcommand{\Regs}{\cP(\Reg)}%{\mathit{Regs}}
\newcommand{\calls}{\,\mathrel{\mathsf{calls}}\,}
\newcommand{\Sig}{\mathit{Sig}}

\newcommand{\Null}{\texttt{Null}}
\newcommand{\Createdat}[1]{\texttt{CreatedAt}\mathopen{}\left( #1 \right)\mathclose{}}
\newcommand{\Union}[2]{\texttt{Union}\mathopen{}\left( #1, #2 \right)\mathclose{}}
\newcommand{\Unknown}{\texttt{Unknown}}

\newcommand{\of}{{:}\,}
\newcommand{\Mid}{\,\mid\,}

\newcommand{\disjoint}[2]{\mathit{disjoint}\left( #1, #2 \right)}

\newcommand{\alg}{\cA}

\begin{document}

\sloppy

\theoremstyle{acmdefinition}
\newtheorem{remark}[theorem]{Remark}

\title[Type-based Enforcement of Infinitary~Trace~Properties for Java]{Type-based Enforcement of\\Infinitary~Trace~Properties for Java}

\author{Serdar Erbatur}
\email{serdar.erbatur@utdallas.edu}
\affiliation{%
  \institution{University of Texas at Dallas}
  \streetaddress{800 W. Campbell Road}
  \city{Richardson}
  \country{USA}
  \postcode{75080-3021}
}

\author{Ulrich Schöpp}
\email{schoepp@fortiss.org}
\affiliation{%
  \institution{fortiss GmbH}
  \streetaddress{Guerickestraße 25}
  \city{Munich}
  \country{Germany}
  \postcode{80805}
}

\author{Chuangjie Xu}
\email{xu@fortiss.org}
\affiliation{%
  \institution{fortiss GmbH}
  \streetaddress{Guerickestraße 25}
  \city{Munich}
  \country{Germany}
  \postcode{80805}
}

\begin{abstract}
A common approach to improve software quality is to use programming guidelines
to avoid common kinds of errors. In this paper, we consider the problem of
enforcing guidelines for Featherweight Java (FJ). We formalize guidelines as
sets of finite or infinite execution traces and develop a region-based type and
effect system for FJ that can enforce such guidelines. We build on the work by Erbatur,
Hofmann and Z\u{a}linescu, who presented a type system for verifying the finite
event traces of terminating FJ programs. We refine this type system,
separating region typing from FJ typing, and use ideas of Hofmann and Chen to
extend it to capture also \emph{infinite} traces produced by
\emph{non-terminating} programs. Our type and effect system can express
properties of both finite and infinite traces and can compute information about
the possible infinite traces of FJ programs. Specifically, the
set of infinite traces of a method is constructed as the greatest fixed point of
the operator which calculates the possible traces of method bodies. Our type
inference algorithm is realized by working with the \emph{finitary abstraction}
of the system based on B\"uchi automata.
\end{abstract}




\maketitle

\pagestyle{plain} % removes running headers

\section{Introduction}

To improve the quality of software, it is desirable to support software
development with automatic static analysis tools. While the full verification
of large-scale software is out of scope for fully automated tools, it is
nevertheless useful to use such tools to identify potential bugs. Current
static analysis tools such as Facebook's infer~\cite{infer} can identify a wide
range of possible bugs, such as memory errors, issues with thread safety or resource
management.

In addition to identifying general bugs, it is also useful to identify logical
errors in the software being developed. Many software development processes
enforce a set of programming guidelines that are designed to rule out
particular kinds of errors. Secure programming guidelines, for example, may
stipulate that all inputs from external sources are to be sanitized before
further processing, that all password changes are to be recorded in a log,
\textit{etc}. While such guidelines typically capture fairly simple program
properties, they help programmers to avoid making typical errors. However,
guidelines need to be enforced to be useful. 

In this paper, we consider the automatic static enforcement of programming guidelines 
for Java. At scale, the fully automatic verification of Java code is out of scope 
of current methods. We therefore concentrate on verifying properties of 
event traces of programs. Due to the relative simplicity of typical guidelines,
this should already be useful for their enforcement.

Let us illustrate the idea using an
example. Consider a server that accepts queries for various 
actions on private data from a user and executes them if the user is authorized 
to do so. At the end, it logs that sensitive data was accessed.
\begin{lstlisting}
  void serve() {
    while(hasQuery()) {
      String query = nextQuery();
      boolean authorized = verifyAuthorization(); // emits $authcheck$;
      if (authorized) {
        readSensitiveData(); // emits $access$;
      }
    }
    LogAccess(); // emits $log$;
  }
\end{lstlisting}
Imagine that functions of interest emit events when they are executed. 
We can then think of a guideline as a set of allowed event traces.
In the example, we have indicated the emitted events in comments. 
An example event trace of
\lstinline|serve| would be $authcheck$ $access$ $authcheck$ $log$. It is
emitted in a run where a user first inputs an authorized request, then
an unauthorized one, and then closes the connection. 
Guidelines, such as ``each access to sensitive data occurs after
authorization'' or ``each access to sensitive data is logged'', would then be 
formalized as ``$access$ must be preceded by~$authcheck$'' and
``each $access$ is eventually followed by $log$'' respectively.

The example is chosen to illustrate that we want to verify properties of
traces not just for terminating runs, but also for non-terminating ones. The
example has intended non-terminating behavior, where the
\texttt{while}-loop is repeated indefinitely. The corresponding traces of
such non-terminating runs are infinite
repetitions of $authcheck$ or $authcheck\; access$. This means we want
to verify that the infinite words in $\{authcheck$, $authcheck\; access\}^\omega$
are allowed by the guideline. In addition to such safety properties (``nothing
bad happens''), one would also like to verify liveness properties (``something
good will happen eventually''). An example for a liveness property would be that any
access must be logged, \ie~that $access$ must be followed by $log$. The above
example code would not have this property, since there are infinite traces
where no access is logged.

There are a number of options of using automatic verification techniques to
verify programming guidelines that have been formalized as sets of allowed
traces. The most obvious option would be to use software model checkers. This
approach should certainly be expressive enough to handle most guidelines, but
it may be limited by scalability problems. One would like to balance
expressiveness and scalability. An alternative is to investigate type systems
and abstract interpretation. The idea is to develop type systems that are
strong enough to express and automatically verify programming guidelines.

% When analyzing infinite traces, one wants to analyze both safety properties
% (``nothing bad happens'') and liveness properties (``something will happen
% eventually''). The example that $access$ must be preceded by $authcheck$ is an
% example of a safety property. An example for a liveness property would be that
% any access must be logged, \ie~that $access$ must be followed by $log$. The
% above example code would not have this property, since  there are infinite
% traces where no access is logged. 

In this paper, we define a type system for Featherweight Java (FJ)~\cite{DBLP:journals/toplas/IgarashiPW01} that allows the
automatic analysis of traces of both terminating and non-terminating runs. In
previous work, Erbatur et al.~\cite{EHZ:enforcing} have developed a similar
type system, which was however limited to finite runs. For infinite runs,
Hofmann and Chen \cite{HC:abstract,HC:buechi} have introduced an abstract interpretation approach for the
type-based analysis of infinite traces. It is based on Büchi automata, shares
their expressiveness and can therefore be used for both safety and liveness
properties. However, it was defined for a simple procedural language. While
\cite{HC:abstract,HC:buechi} contain a sketch of how to apply their methods in
a type system like~\cite{EHZ:enforcing}, this development remains at the level
of a sketch that needs to be fully developed in further work%
% TODO: revise
\footnote{Hofmann and Chen~\cite{HC:abstract} do not yet formulate a soundness theorem and the 
sketched notion of well-typedness is not fully precise
(the formulation is unclear w.r.t.~the quantification of~$\eta$).
In addition, some choices in the sketched type system seem unnecessarily limiting: For example, 
any method that calls another method from the same class would be analyzed as
having a non-terminating run even when there is no recursion at all.}.

Here we develop a type system for FJ on the basis of ideas from both
approaches. In contrast to previous work like~\cite{EHZ:enforcing,HC:abstract},
our type system is not a refinement of the FJ type system, but a separate flow
type system in the spirit of Microsoft's TypeScript~\cite{typescript} and
Facebook's Flow~\cite{flow}. The resulting type system is simpler, more
precise, and more flexible than a refinement of FJ's type system such as those in~\cite{EHZ:enforcing,HC:abstract}.

Our development is based on three main ideas: effect types, regions and abstraction.

\textbf{Effect Types.}
Type systems can be extended to statically approximate sets of possible event
traces by means of effect types. The idea is to extend the type with effect
annotations that give information about the possible traces. This approach was
followed in~\cite{EHZ:enforcing}, where method types have the form
\lstinline[keywordstyle=]|void serve()$ \andrel A$|, with an \emph{effect annotation}
``$\&\, A$'' that gives static information about the possible traces of
terminating runs of \texttt{serve}. In essence, $A$ represents a set of traces
that includes all possible traces of terminating runs. In our example,
 $A$ could be $\{authcheck, authcheck\, access\}^*\cdot \{log\}$, 
 that is, the set of traces that are a finite repetition of
$authcheck$ or $authcheck\, access$, followed by $log$.

\textbf{Regions.}
To obtain a useful type system for the analysis of traces in Java, the type
system needs to perform some kind of flow analysis. Consider, for example, the
Java interface \lstinline|Runnable| with a single method 
\lstinline[keywordstyle=]|void run()|. If we try to annotate this method with an effect annotation for the set
of its possible traces, then, without further knowledge, we need to include the
traces of all classes that implement \lstinline|Runnable|. The effect typing of
a call like \lstinline|x.run()| would not be very useful, since it would include the
effects of all implementations of \lstinline|run()|. What is needed is a way to
narrow down the possible objects that~\texttt{x} can point to. This is
why~\cite{EHZ:enforcing} integrates both regions and effects in their type
system. We also use regions, but diverge from~\cite{EHZ:enforcing} in that we
capture them in a separate type system rather than integrating them in the Java
type system. Our type system is simpler, more precise and can accommodate richer
kinds of regions.

\textbf{Abstraction.}
In an implementation of automatic type checking it is of course not possible to
store infinite languages like~$A$ above directly. Some computable abstraction
is needed, a natural choice being finite automata. This choice has been
taken for finite words in~\cite{EHZ:enforcing} and for potentially infinite 
words in~\cite{HC:abstract,HC:buechi}. In particular, Hofmann and 
Chen~\cite{HC:abstract} define an abstraction of Büchi automata for the finite
representation of infinite words. 
Our type system works with the same abstractions.
However, the particular choice of abstraction is not important for it.
We introduce a notion of Büchi abstraction and establish soundness of the 
type system for any instance of it. The automata-theoretic 
constructions of~\cite{EHZ:enforcing,HC:abstract} provide an example.

In summary, the contributions of this paper include:
\begin{itemize}
\item We introduce a new effect type system that is simpler than those of~\cite{EHZ:enforcing,HC:abstract} but allows more precise analysis. It can analyze both terminating and non-terminating programs.
\item We give a simple characterization of the infinitary effect analysis as a single greatest fixed point, which allows us to prove its soundness simply by induction on the approximation of the greatest fixed point.
\item We realize the type inference by working with the finitary abstraction of the type system following~\cite{HC:abstract}. To have a generic formulation, we introduce a notion of B\"uchi abstraction that abstracts from the automata-theoretic constructions of~\cite{HC:abstract}.
\end{itemize}
Moreover, we have a prototype implementation of type inference based on Soot~\cite{Soot}.

\section{Events and Traces}
\label{sec:lang}

Let $\Sigma$ be a finite alphabet. We use $\Sigma^*$ to denote the set of finite words over $\Sigma$, and $\Sigma^\lo$ the set of finite and infinite words. We write $uv$ for the concatenation of the finite word $u$ with the word $v$ which can finite or infinite. We now extend concatenation to and recall some other operations on languages: Let $U,U' \subseteq \Sigma^*$ and $V \subseteq \Sigma^\lo$. We write
\begin{itemize}
\item $U \cdot U' = \{ uu' \mid u \in U, u' \in U' \} \subseteq \Sigma^*$,
\item $U \cdot V = \{ uv \mid u \in U, v \in V \} \subseteq \Sigma^\lo$,
\item $U^* = \{ u_1u_2\ldots u_n \mid u_i \in U \text{ for all } i\in\{1,\ldots,n\} \} \subseteq \Sigma^*$,
\item $U^+=U^*\setminus\{\varepsilon\} \subseteq \Sigma^*$, where $\varepsilon$ is the empty word, and
\item $U^\omega = \{ u_1u_2\ldots u_i \ldots \mid u_i \in U \text{ for all } i \in \N \} \subseteq \Sigma^\lo$.
\end{itemize}
We often write $a^*$ rather than $\{a\}^*$ and similarly $a^+,a^\omega$. Note that $U^*$ is the least fixed point of $\lambda X. \{\varepsilon\} \cup U \cdot X$. However, $U^\omega$ may not be the greatest fixed point of $\lambda X.U\cdot X$, because $U$ may contain the empty word. For instance, $\{\varepsilon\}^\omega = \{\varepsilon\}$ while the greatest fixed point of $\lambda X. \{\varepsilon\} \cdot X$ is $\Sigma^\lo$. In fact, $U^\omega$ is the greatest fixed point of $\lambda X. L \cup (U\setminus\{\varepsilon\}) \cdot X$, where $L= U^*$ if $\varepsilon \in U$ and $L = \emptyset$ otherwise.

For the sake of simplicity, we assume programs have special commands issuing \emph{events} from $\Sigma$. Then a terminating execution of a program will generate a finite \emph{trace}, that is, a word from $\Sigma^*$, whereas a non-terminating execution will generate a finite or infinite trace, that is, a word from $\Sigma^\lo$. Our objective is to capture all possible traces of a given program via a type and effect system, so that we can verify whether they are allowed by a given \emph{guideline}, that is, a language of acceptable traces.

\section{Featherweight Java}
\label{sec:FJ}

We work with a variant of Featherweight Java (FJ) that extends the FJ calculus of~\cite{DBLP:journals/toplas/IgarashiPW01} with filed updates. Following the formulation of~\cite{EHZ:enforcing}, we add primitive if- and let-expressions for convenience and omit constructors for simplicity. Moreover, we add a primitive operation $\tto(a)$ that generates the event $a$ from a finite alphabet~$\Sigma$.

The syntax of the language uses four kinds of names:
\[
\begin{array}{rlcrl}
\text{variables: } & x,y \in \Var,
& \qquad &
\text{classes: } & C,D \in \Cls,
\\
\text{fields: } & f \in \Fld,
& &
\text{methods: } & m \in \Mtd.
\end{array}
\]
Program expressions are defined as follows:
\[
\begin{aligned}
\Expr \ni e ::=  \ x &\Mid \letin{x}{e_1}{e_2} \Mid \ttifthenelse{x}{y}{e_1}{e_2} \\
& \Mid \ttnull \Mid \new{\ell}{C} \Mid \cast{C}{e} \Mid \tto(a) \\
& \Mid x^C.m(\bar{y}) \Mid x^C.f \Mid x^C.f := y
\end{aligned}
\]
The expression $\new{\ell}{C}$ is annotated with a label~$\ell$. We use 
labels only to distinguish different occurrences of $\mathtt{new}$ in a program;
since our type system will track where objects were created.
In a few expressions we have added type annotations and write $x^C$ for a
variable of type~$C$. They will be needed when looking up in the class table (see rules \textsc{get}, \textsc{set} and \textsc{call} in Figure~\ref{fig:typing}).
We sometimes omit annotations when they are not needed.

We assume three distinguished formal elements: $\this \in \Var$, $\Object\in \Cls$ and $\NullType \in \Cls$. The $\NullType$ class plays the role of the type of $\ttnull$ from the Java language specification~\cite[\S4]{JavaSpec}. It may not be used in programs, \ie~we require $C \not= \NullType$ in the expressions $\new{\ell}{C}$ and $\cast{C}{e}$. When $x$ is not a free variable of $e_2$, we may write $e_1;e_2$ rather than $\letin{x}{e_1}{e_2}$.

An FJ program $(\prec,\fields,\methods,\mtable)$ consists of
\begin{itemize}
\item a subtyping relation ${\prec} \in \cPf(\Cls \times \Cls)$ with $C \prec D$ meaning that $C$ is an immediate subclass of $D$,
\item a field list $\fields: \Cls \to \cPf(\Fld)$ mapping a class to its fields,
\item a method list $\methods: \Cls \to \cPf(\Mtd)$ mapping a class to its methods, and
\item a method table $\mtable: \Cls \times \Mtd \rightharpoonup \Var^* \times \Expr$ mapping a method to its formal parameters and its body.
\end{itemize}
All components are required to be \emph{well-formed} (see \eg~\cite[Section~3]{EHZ:enforcing} for details). Let $\preceq$ be the reflexive and transitive closure of $\prec$. We have $C\preceq \Object$ and $\NullType \preceq C$ for all $C\in \Cls$.

\begin{example}
\label{example:code}
As a running example, we consider the following Java code,
in which ${\ell_1}$, ${\ell_2}$ and ${\ell_3}$ are arbitrary fresh labels.
\begin{center}
\vspace{-\abovedisplayskip}
\noindent
\begin{minipage}[t]{.21\textwidth}
\begin{lstlisting}
class Node {
 Node next;
 Node last() {
  emit($a$);
  if (next == null) {
   return this;
  } else {
   return next.last();
  }
 }
}
\end{lstlisting}
\end{minipage}
\hspace{10pt}
\begin{minipage}[t]{.228\textwidth}
\begin{lstlisting}
class Test {
 Node linear() {
  Node x = new$^{\ell_1}$ Node();
  Node y = new$^{\ell_2}$ Node();
  y.next = x;
  return y.last();
 }
 Node cyclic() {
  Node z = new$^{\ell_3}$ Node();
  z.next = z;
  return z.last();
 }
}
\end{lstlisting}
\end{minipage}
\end{center}
In Featherweight Java, the class \Node{} is represented by
\[
\begin{aligned}
\fields(\Node) = \ & \{\texttt{next}\} \\
\methods(\Node) = \ & \{\ttlast\} \\
\mtable(\Node, \ttlast) = \ & ((), e_\ttlast) \\
\end{aligned}
\]
where the expression $e_\ttlast$ is defined as follows:
\[
\begin{aligned}
\hspace*{20pt}
e_\ttlast := \ & \letin{\_}{\tto(a)}{} \\
& \letin{x}{\this^\Node.\texttt{next}}{} \\
& \letin{y}{\ttnull}{} \\
& \texttt{if} \; x = y \; \texttt{then} \; \this \\
& \texttt{else} \; \letin{z}{\this^\Node.\texttt{next}}{z^\Node.\ttlast()} \hspace*{10pt} \qed
\end{aligned}
\]
\end{example}

In the standard FJ type system~\cite{DBLP:journals/toplas/IgarashiPW01}, types are simply classes. In the rest of this paper, we consider only well-typed FJ programs. The type information of an FJ program can be modeled by a class table $(F_0,M_0)$, where $F_0 : \Cls \times \Fld \rightharpoonup \Cls$ is a \emph{field typing} that assigns to each class $C$ and each field $f \in \fields(C)$ the class of~$f$, and $M_0 : \Cls \times \Mtd \rightharpoonup \Cls^* \times \Cls$ a \emph{method typing} that specifies for each class $C$, method $m \in \methods(C)$ and classes of the argument variables the class of the result value of $m$.

The operational semantics of FJ uses a notion of state~$(s,h)$ that consists of a \emph{store} $s : \Var \rightharpoonup \Val$ mapping variables to values and a \emph{heap} $h : \Loc \rightharpoonup \Obj$ mapping locations to objects. The only kinds of values are object locations and $\itnull$. An object $(C,G,\ell)\in \Obj$ contains a class identifier $C \in \Cls$, a valuation $G : \Fld \rightharpoonup \Val$ of its fields and a label $\ell \in \Pos$ indicating where it was created (\ie, the object was created by a $\new{\ell}{C}$ expression).
\[
\begin{aligned}
\text{locations:} && l & \in \Loc \\
\text{values:} && v & \in \Val = \Loc \uplus \{\itnull\} \\
\text{stores:} && s & \in \Var \rightharpoonup \Val \\
\text{heaps:} && h & \in \Loc \rightharpoonup \Obj \\
\text{objects:} && (C,G,\ell) & \in \Obj = \Cls \times (\Fld \rightharpoonup \Val) \times \Pos
\end{aligned}
\]
The operational semantics of terminating evaluations is given as a big-step relation $(s,h) \vdash e \Downarrow v,h' \andrel w$. It expresses that, in the state $(s,h)$, the expression $e$ evaluates to the value $v$ with the heap updated to $h'$, generating the event trace $w \in \Sigma^*$. For diverging evaluations, we work with a coinductively generated relation $(s,h) \vdash e \Uparrow \andrel w$ meaning that the evaluation of~$e$ in the state $(s,h)$ diverges and generates an event trace $w \in \Sigma^\lo$. The operational semantics rules are standard and can be found in Appendix~\ref{sec:operational}.

\section{Region Types and Effects}
\label{sec:type:system:sets}

We introduce a type system for region and effect analysis for FJ. In contrast
to previous work~\cite{EHZ:enforcing,HC:abstract}, we do not refine FJ type
system, but instead build an independent region typing system for FJ. Region
type information is complementary to FJ type information and can be captured
without repeating the FJ type system. This also makes sense practically: we
would not want to modify the Java compiler but develop an additional analysis 
system.

For the sake of understandability, in this section we present the
system with arbitrary languages over~$\Sigma$ as effect annotations. In 
Section~\ref{sec:buechi} we shall make it algorithmic by moving to a finitary
abstraction of languages.

\subsection{Region Types}
\label{sec:region}

Our type system captures region information. A \emph{region} represents
provenance information about a value, such as where in the program an object 
was created or what its actual class is. The intention is that the set of 
values is subdivided in regions according to various criteria.

In this paper, we use the following definition of regions:
\[
\Reg \ni r, s\ ::=\ \Null \Mid \Createdat{\ell} \Mid \Unknown
\]
These regions have the following meaning. The region $\Null$ contains only the value $\itnull$. The region $\Createdat{\ell}$ contains all references to objects that were created by an expression of the form $\new{\ell}{C}$. This region allows us to track where in the program an object originates. The region $\Unknown$ is for references to objects of unknown origin, \eg~from library code. A formal interpretation of regions is given in Section~\ref{sec:soundness}.

It is possible to use a richer definition of regions to capture other properties of interest, such as
taintedness, or to give regions more structure. For example, one may allow
unions or intersections of regions. We plan to do that in further work. Here we use
a simple representative definition of regions in order to focus on the new type
system itself and on infinitary traces.

Regions provide information about the actual classes of objects.
We write $\Cls(r)$ for the set of possible classes of an object in region~$r$,
which is defined as follows:
$\Cls(\Null) = \emptyset$, 
$\Cls(\Createdat{\ell}) = \{C \mid \text{program contains $\new{\ell}{C}$}\}$,
and
$\Cls(\Unknown) = \Cls$.
Note that that $\Cls(r)$ is \emph{not} required to be closed under superclasses.
Indeed, $\Cls(\Createdat{\ell})$ will be just a singleton if the label $\ell$
was used only once.

Regions also provide information about the identity of objects. Objects in
disjoint regions cannot be identical. The region $\Null$ and all regions of
the form $\Createdat{\ell}$ are all pairwise disjoint. However, $\Unknown$ 
overlaps with all other regions. We write $\mathit{disjoint}(r, s)$
to indicate that two regions are disjoint.

\subsection{Region and Effect Expressions}
\label{sec:formal:expression}

Before diving into the details of the type system, we explain the region
and effect expressions that are used in the type system.
The typing judgement will have the following form:
\[
  \Gamma \vdash e : T \calls S.
\]
The context~$\Gamma$ simply maps program variables to regions. 
More interesting are~$T$ and~$S$ on the right-hand side. They both are 
expressions that contain region and effect information about the computation
of~$e$. In this section we explain the meaning of these expressions.

First, $T$ is an expression of the form
$r_1 \andrel U_1 \mid \dots \mid r_n \andrel U_n$.
It lists the possible options for the terminating computations of~$e$.
The option $r_i \andrel U_i$ means that the result value is in region~$r_i$ and
the computation trace is in~$U_i\subseteq \Sigma^*$. Thus, $T$ expresses that whenever~$e$ 
evaluates to a value, there will be some~$i$ such that the result is in 
region~$r_i$ and the computation trace is in~$U_i$.

For example, the following expression
\[
\begin{aligned}
& \mathtt{let}\ x = \ttifthenelseg{\mathit{cond}}{(\tto(a);\, \new{\ell_1}{C})}{(\new{\ell_2}{D})}\\
& \mathtt{in}\ \tto(b);\, x
\end{aligned}
\]
will have type 
$\Createdat{\ell_1} \andrel \{ab\} \mid \Createdat{\ell_2} \andrel \{b\}$.
It either evaluates to a reference pointing to an object that was created 
by a \texttt{new} with label~$\ell_1$ or to one that was created by a \texttt{new} with 
label~$\ell_2$. In the first case, the trace of the computation is $ab$, 
and in the latter case it is just $b$. 
The expression may also be given the less precise type
$\Unknown \andrel \{ab, b\}$, which just expresses that it produces some value
and either has effect~$ab$ or~$b$ while doing so.

Second, $S$ is a \emph{call expression} of the form 
$U_1 \cdot \delta_1 \cup \dots \cup U_k \cdot \delta_k$.
It contains information about the calls that~$e$ may make.
In the expression, each $\delta_i$ is a method signature, which is given
by a class name, a method name and region information. 
We specify signatures precisely below, for a first explanation of~$S$ 
the details are not important.
The expression~$S$ gives us information about which method calls~$e$ can make
and what the traces leading up to such calls can be. It expresses
that for any call that~$e$ makes,
there exists an~$i$ such
the call goes to a method with signature~$\delta_i$
and that the event trace leading to the call is in~$U_i$.

Consider, for example, the following FJ-expression~$e$:
\[
  \tto(a);\, x.\mathtt{f}();\, \tto(b);\, x.\mathtt{g}()
  \enspace.
\]
If we assume that $\mathtt{f}$ and $\mathtt{g}$ have no effects on their own,
then the call expression~$S$ for~$e$ would be 
$\{a\}\cdot \delta_{\mathtt{f}} \cup \{ab\}\cdot \delta_{\mathtt{g}}$,
where ${\delta_{\mathtt{f}}}$ and $\delta_{\mathtt{g}}$ are the signatures
of~$\mathtt{f}$ and~$\mathtt{g}$ respectively.

The information about possible method calls will be useful to approximate the traces
of non-terminating computations. A non-terminating computation in Featherweight
Java must consist of an infinite series of method calls. To compute an
approximation of all infinite traces of such sequences, it suffices to compute
the finite effect traces from one method call to the next and to consider their
concatenation in all possible infinite sequences of method calls. 
We do this by computing an expression like~$S$ for the
body of each method (Section~\ref{sec:system}). Then we compute an approximation of all infinite traces
from this information (Section~\ref{sec:infinitary}).

To work with the expressions~$T$ and~$S$, we need some notation. While we use
different notation for both kinds of expressions to emphasize their different
meaning, they are in fact instances of the same kind of formal
expressions.

\begin{definition}[Formal effect expression]
\label{def:formal:expr}
Given a set~$K$ of keys, we define the set
$\cP(\Sigma^*)\langle K\rangle$ of \emph{formal effect expressions} 
to be the set of finite partial functions 
from~$K$ to~$\cP(\Sigma^*)$. We write $\emptyset$ to denote the empty expression.
For $T, T' \in \cP(\Sigma^*)\langle K\rangle$ and $U\in \cP(\Sigma^*)$, we define:
\begin{itemize}
\item $T \subseteq T'$ if and only if\/ $T(x) \subseteq T(x')$ for all $x \in K$,
\item $T \cup T' \in \cP(\Sigma^*)\langle K\rangle$ by $(T \cup T')(x) = T(x) \cup T'(x)$,
\item $U\cdot T \in \cP(\Sigma^*)\langle K\rangle$ by $(U \cdot T)(x) = U \cdot T(x)$,
\end{itemize}
where $T(x)$ is treated as the empty set if $x \not\in \dom(T)$.
\end{definition}

We use the notation $r_1 \andrel U_1 \mid \dots \mid r_n \andrel U_n$ to represent the elements of $\cP(\Sigma^*)\langle \Reg\rangle$. Such an expression denotes the function mapping~$r_i$ to $U_i$ for $i=1,\dots, n$.
For example, we have
$(r_1 \andrel U_1 \mid r_2 \andrel U_2) \cup (r_1 \andrel U'_1 \mid r_3 \andrel U_3)
=(r_1 \andrel (U_1 \cup U'_1) \mid r_2 \andrel U_2 \mid r_3 \andrel U_3)$ and
$U \cdot (r_1 \andrel U_1 \mid r_2 \andrel U_2)
=(r_1 \andrel (U\cdot U_1) \mid r_2 \andrel (U\cdot U_2))$.
Likewise, call expressions
$U_1\cdot \delta_1 \cup \dots \cup U_n\cdot \delta_n$
are elements of $\cP(\Sigma^*)\langle \Sig \rangle$, where $\Sig$ is a set of method signatures (to be defined below). When $U_i = \{ \varepsilon \}$ we may omit it and simply write $\delta_i$ in the expressions.

\subsection{Typing Rules}
\label{sec:system}

Having explained region and effect expressions, we can now define the type system in detail. As for Featherweight Java, we need a \emph{class table} to record the region types and effect expressions of methods and fields. This is needed to formulate typing rules for method call and field access.

We call a tuple $(C,r,m,\bar{s})$ of a class $C$, a region $r$, a method $m \in \methods(C)$ and a sequence $\bar{s}$ of regions for the arguments of $m$ a \emph{method signature}. We use $\Sig \subseteq \Cls \times \Reg \times \Mtd \times \Reg^*$ to denote the set of method signatures.

\begin{definition}[Class table]
\label{def:class:table}
A \emph{class table}~$(F,M)$ consists of
\begin{itemize}
\item a \emph{field typing} $F : \Cls \times \Reg \times \Fld \rightharpoonup \Regs$ that assigns to each class $C$, region $r$ and field $f \in \fields(C)$ a set $F(C,r,f)$ of possible regions of the field~$f$, and
\item a \emph{method typing} $M : \Sig \to \cP(\Sigma^*)\langle\Reg\rangle \times \cP(\Sigma^*)\langle\Sig\rangle$ that assigns to each method signature a pair of an expression for regions and effects of terminating executions and a call expression.
\end{itemize}
A class table is required to be \emph{well-formed} in the following sense:
\begin{itemize}
\item $\mathtt{Null} \in F(C,r,f) \subseteq F(C,\mathtt{Unknown},f)$;
\item $F(C,r,f) = F(D,r,f)$ whenever $C \preceq D$ and $f \in \fields(D)$;
\item $M(C,r,m,\bar{s}) \subseteq M(D,r,m,\bar{s})$ whenever $C \preceq D$ and $m \in \methods(D)$, where the order $\subseteq$ on $\cP(\Sigma^*)\langle\Reg\rangle \times \cP(\Sigma^*)\langle\Sig\rangle$ is defined componentwise.
\end{itemize}
\end{definition}
The well-formedness conditions reflect the subtyping properties of FJ. We additionally require all entries of $F$ to contain the region $\Null$, because the fields of newly created objects are all initialized to $\itnull$. We require~$F$ to be invariant w.r.t.\ also the ``largest'' region $\Unknown$. Note that the tables may contain impossible entries whose class and its subclasses are not in the region. They are not important in the typing inference. We can set them to the ``bottom''. A typical example is $M(C,\Null,m,\bar{s}) = (\emptyset,\emptyset)$.

We denote $F(C,r,f) = R$ and $M(C,r,m,\langle s_1, \ldots, s_n \rangle) = (T,S)$ as
\begin{lstlisting}
class ${\color{black}C}@r$
  ${\color{black}f}$ : $R$
  ${\color{black}m}$($s_1$,...,$s_n$) : $T \calls S$ 
\end{lstlisting}
with which we hope to improve readability.

The typing judgments take the form $\Gamma \vdash e : T \calls S$, where $\Gamma: \Var \rightharpoonup \Reg$ is a typing environment, $e \in \Expr$ a term expression, $T\in \cP(\Sigma^*)\langle\Reg\rangle$ an expression for regions and terminating effects, and $S \in \cP(\Sigma^*)\langle\Sig\rangle$ an expression for call effects.

\begin{figure}[h]
\[
\inferrule*[left={sub}]
{\Gamma \vdash e : T \calls S\quad
 T \subseteq T'\quad
 S \subseteq S'}
{\Gamma \vdash e : T' \calls S'}
\]
\[
\inferrule*[left={weak}]
{\Gamma \vdash e : (r \andrel U) \cup T \calls S }
{\Gamma \vdash e : (\Unknown \andrel U) \cup T \calls S}
\]
\[
\inferrule*[left={prim}]
{\ }
{\Gamma \vdash \tto(a) : \Null \andrel \{a\} \calls \emptyset}
\]
\[
\inferrule*[left={var}]
{\ }
{\Gamma,\, x\of r  \vdash x : r \andrel \{\varepsilon\} \calls \emptyset}
\]
\[
\inferrule*[left={if}]
{\Gamma,\, x\of r,\, y\of s \vdash e_1 : T_1 \calls S_1 \\
 \Gamma,\, x\of r,\, y\of s \vdash e_2 : T_2 \calls S_2}
{\Gamma,\, x\of r,\, y\of s \vdash \ttifthenelse{x}{y}{e_1}{e_2} : T_1 \cup T_2 \calls S_1 \cup S_2}
\]
\[
\inferrule*[left={else}]
{\mathit{disjoint}(r, s) \\
 \Gamma,\, x\of r,\, y\of s \vdash e_2 : T_2 \calls S_2}
{\Gamma,\, x\of r,\, y\of s \vdash \ttifthenelse{x}{y}{e_1}{e_2} : T_2 \calls S_2}
\]
\[
\inferrule*[left={let}]
{\Gamma \vdash e_1 : r_1 \andrel U_1 \mid \ldots \mid r_n \andrel U_n \calls S\\
 \Gamma,\, x\of r_i \vdash e_2 : T_i \calls S_i}
{\textstyle
 \Gamma \vdash \letin{x}{e_1}{e_2} : \bigcup^n_{i=1} U_i \cdot T_i \calls S \cup \bigcup^n_{i=1} U_i \cdot S_i}
\]
\[
\inferrule*[left={null}]
{\ }
{\Gamma \vdash \ttnull : \Null \andrel \{\varepsilon\} \calls \emptyset}
\]
\[
\inferrule*[left={new}]
{\ }
{\Gamma \vdash \new{\ell}{C} : \Createdat{\ell} \andrel \{\varepsilon\} \calls \emptyset}
\]
\[
\inferrule*[left={cast}]
{\Gamma \vdash e : T \calls S}
{\Gamma \vdash \cast{D}{e} : T \calls S}
\]
\[
\inferrule*[left={get}]
{\textstyle T = \bigcup \{ s \andrel \{\varepsilon\} \mid s \in F(C,r,f) \}}
{\Gamma,\, x\of r \vdash x^C.f : T \calls \emptyset}
\]
\[
\inferrule*[left={set}]
{s \in F(C,r,f)}
{\Gamma,\, x\of r,\, y\of s \vdash x^C.f := y : s \andrel \{\varepsilon\} \calls \emptyset}
\]
\[
\inferrule*[left={call}]
{(T,\_) = M(C,r,m,\bar{s})}
{\Gamma,\, x\of r ,\, \bar{y}\of \bar{s} \vdash x^C.m(\bar{y}) : T \calls (C, r, m, \bar{s})}
\]
\caption{Region-based type and effect system}
\label{fig:typing}
\end{figure}

The typing rules are given in Figure~\ref{fig:typing}.
Rule \textsc{weak} says that any possible region $r$ can be ``weakened'' to \Unknown{}. It allows us to merge the sets of possible traces. For instance, if $e : r \andrel U \mid \Unknown \andrel U'$ then by \textsc{weak} we have $e : \Unknown \andrel (U \cup U')$.
Rule \textsc{prim} says that $\tto(a)$ has the terminating trace~$a$ and no calls. Rules \textsc{null}, \textsc{new} and \textsc{get} have similar effects.
Rule \textsc{else} says that if $x$ and $y$ are in disjoint regions, then the whole statement has the same type as the else-branch, because objects in disjoint regions cannot be identical. This illustrates how the usage of regions allows the analysis to achieve some path-sensitivity.
In rule \textsc{let}, $e_1$ may have multiple regions and effects, so we type $e_2$ with $x$ in each of these regions $r_i$ to get the effects of $e_2$ and then take the union of them. In rule \textsc{call} for method call, we simply look up the method typing~$M$ for the terminating effect and take the method signature as the call expression. 

\begin{example}[Expression Typing]
\label{firstexample}
Suppose that there are two classes $D \prec C$ with two methods $f$ and $g$. Consider the following class table:
\begin{center}
\vspace{-\abovedisplayskip}
\noindent
\begin{minipage}[t]{.24\textwidth}
\begin{lstlisting}
class ${\color{black}C}@{\Createdat{\ell_1}}$
  ${\color{black}f}$() : $\Null \andrel \{a\} \calls \dots$
  ${\color{black}g}$() : $\Null \andrel \{aa\} \calls \dots$
\end{lstlisting}
\end{minipage}
\begin{minipage}[t]{.23\textwidth}
\begin{lstlisting}
class ${\color{black}D}@{\Createdat{\ell_2}}$
  ${\color{black}f}$() : $\Null \andrel \{b\} \calls \dots$
  ${\color{black}g}$() : $\Null \andrel \{bb\} \calls \dots$
\end{lstlisting}
\end{minipage}
\end{center}
Following the typing rules, the FJ-expression
\[
\ttifthenelseg{\mathit{cond}}{(\new{\ell_1}{C})}{(\new{\ell_2}{D})}
\]
has type $\Createdat{\ell_1} \andrel \{\varepsilon\} \mid \Createdat{\ell_2} \andrel \{\varepsilon\}$. Write $e$ for the above expression. Now consider the following expressions:
\[
\letin{x}{e}{x.{f}()}
\ \ \
\letin{x}{e}{x.{f}();x.{f}()}
\ \ \
\letin{x}{e}{x.{g}()}
\]
The first expression has type $\Null \andrel \{a,b\}$, the second and third one have type $\Null \andrel \{aa,bb\}$. One may have expected the traces of the second expression to be $\{a,b\} \cdot \{a,b\}$, that is, $\{aa,ab,ba,bb\}$, because~$x$ may be in two different regions. Indeed, this is the case in the system of~\cite{EHZ:enforcing}, but not in ours. In the \textsc{let} rule, the body $x.{f}();x.{f}()$ is typed twice with the two regions of $x$, and the results $\{aa\}$ and $\{bb\}$ are joined in the end, resulting in~$\{aa,bb\}$.

This seems to be a slight defect of~\cite{EHZ:enforcing}. Indeed, the function $g$ could have been defined to have body
$\this.f();\this.f()$. In this case, the second expression would be the inlining of the third and one would perhaps expect inlining
not to lose information.
\qed
\end{example}

\begin{definition}[Well-typedness]
\label{def:well-typedness}
An FJ program $(\prec, \fields, \methods$, $\mtable)$ is \emph{well-typed} w.r.t.\ a class table $(F,M)$ if, for each method signature $(C,r,m,\bar{s}) \in \Sig$ with $C \in \Cls(r)$ and $M(C,r,m,\bar{s}) = (T,S)$ and $\mtable(C,m)=(\bar{x},e)$, the judgment $\mathtt{this} \of r,\, \bar{x} \of \bar{s} \vdash e : T \mathrel{\mathsf{calls}}S$ is derivable.
\end{definition}
The program from Example~\ref{example:code} is well-typed w.r.t.\ the class table in Example~\ref{example:typing}. Another example of a well-typed program is given in Example~\ref{example:subtyping}.

One subtle aspect of well-typedness is the condition $C \in \Cls(r)$. This
condition allows us to use the region information to statically narrow down the
possible targets of calls. This is an improvement over~\cite{EHZ:enforcing},
where the region information was not used. With the condition, well-typedness
imposed no requirement at all on the method table entry $M(C,r,m,\bar{s})$ when
$C\notin\Cls(r)$. Informally, this is because the method table entry belongs to
a method that, in region~$r$, can never be executed.

The following example illustrates how our type system makes use of region information to narrow down the possible targets of call expressions. 

\begin{example}[Regions and Subtyping]
  \label{example:subtyping}
  Consider a program with two simple classes (the return type is \lstinline{Object} only because we do not have a \lstinline{void} type):
\begin{center}
\vspace{-\abovedisplayskip}
\begin{minipage}[t]{.2\textwidth}
\begin{lstlisting}
class A {
  Object f() {
    emit($a$);
    return null;
  }
}
\end{lstlisting}
\end{minipage}
\qquad
\begin{minipage}[t]{.19\textwidth}
\begin{lstlisting}
class B extends A {
  Object f() {
    emit($b$);
    return null;
  }
}
\end{lstlisting}
\end{minipage}
\end{center}
Let $\ell_1$ and $\ell_2$ be labels such that
$\Cls(\Createdat{\ell_1}) = \{\mathtt{A}\}$ and
$\Cls(\Createdat{\ell_2}) = \{\mathtt{B}\}$.

The program is well-typed w.r.t.\ the following class table:
\begin{center}
\vspace{-\abovedisplayskip}
\begin{minipage}[t]{.21\textwidth}
\begin{lstlisting}
class A$@{\Unknown}$
  f(): $\Null \andrel \{a,b\} \calls \emptyset$
\end{lstlisting}
\end{minipage}
\qquad
\begin{minipage}[t]{.2\textwidth}
\begin{lstlisting}
class B$@{\Unknown}$
  f(): $\Null \andrel \{b\} \calls \emptyset$
\end{lstlisting}
\end{minipage}

\begin{minipage}[t]{.21\textwidth}
\begin{lstlisting}
class A$@{\Createdat{\ell_1}}$
  f(): $\Null \andrel \{a\} \calls \emptyset$
\end{lstlisting}
\end{minipage}
\qquad
\begin{minipage}[t]{.2\textwidth}
\begin{lstlisting}
class B$@{\Createdat{\ell_1}}$
  f(): $\emptyset \calls \emptyset$
\end{lstlisting}
\end{minipage}

\begin{minipage}[t]{.21\textwidth}
\begin{lstlisting}
class A$@{\Createdat{\ell_2}}$
  f(): $\Null \andrel \{b\} \calls \emptyset$
\end{lstlisting}
\end{minipage}
\qquad
\begin{minipage}[t]{.2\textwidth}
\begin{lstlisting}
class B$@{\Createdat{\ell_2}}$
  f(): $\Null \andrel \{b\} \calls \emptyset$
\end{lstlisting}
\end{minipage}
\end{center}
The types in region~$\Unknown$ should be unsurprising. The requirement that
the typings be closed under superclasses requires us to include the effect~$b$
also in class~$A$.

In region $\Createdat{\ell_1}$, the entry for~\lstinline{f} in class~\lstinline{A}
should be unsurprising. The entry for~\lstinline{f} in~\lstinline{B} can
be explained as follows. Because of 
$\Cls(\Createdat{\ell_1}) = \{\mathtt{A}\}$, we know that all objects in 
region $\Createdat{\ell_1}$ have class~\lstinline{A}. This means that the
code for~\lstinline{f} in~\lstinline{B} will never be executed and we
may assume the method to have empty effect.
The type system supports this reasoning because the definition of well-typedness
for programs makes a requirement only for entries $(C,r,m,\bar{s})$ with $C\in\Cls(r)$.
Accordingly, we may take the empty effect for the method table entry
for~\lstinline{f} in~\lstinline{B}.  

In region $\Createdat{\ell_2}$, it is the entry for~\lstinline{f} in~\lstinline{A}
that needs explanation. It is reasonable because 
$\Cls(\Createdat{\ell_1}) = \{\mathtt{B}\}$ tells us that the actual class of 
any object in region $\Createdat{\ell_2}$ is~\lstinline{B}, so that any 
call of~\lstinline{f} would go to the code from class~\lstinline{B}.
In the method table, we cannot take the empty effect for~\lstinline{f} in~\lstinline{A},
however, as the table needs to be closed under superclasses. We need to
include at least the effects of \lstinline{B.f} in those of \lstinline{A.f}.

With the above class table, consider the following code fragment.
\begin{lstlisting}
  A x = new$^{\ell_2}$ B();
  x.f();
\end{lstlisting}
The type system allows us to give~\lstinline{x} type $\Createdat{\ell_2}$. In
this case, the method call \lstinline{x.f()} is analyzed to have
effect~$\{b\}$, which is exact. We could also give~\lstinline{x} the
type $\Unknown$, but this choice would lead to an over-approximation of the
actual effect. We would get that the method call \lstinline{x.f()} has
effect~$\{a, b\}$. However, we would not be allowed to give~\lstinline{x} the
type $\Createdat{\ell_1}$.
  \qed
\end{example}

\subsection{Infinite Traces}
\label{sec:infinitary}

We now turn to analyzing possibly infinite traces of non-terminating programs. These cannot be captured directly from an inductive typing derivation. Instead, we use the type system to justify the call expressions (see Definition~\ref{def:well-typedness} of well-typed programs) and subject them the following \emph{coinductive} treatment which will lead to the correctness of the non-terminating effects (Theorem~\ref{thm:soundness:gen}).

\begin{definition}[Infinitary effect typing]
Given a method typing $M$, we write $S_\delta$ to denote the second components of $M(\delta)$. 
We call an assignment $\eta : \Sig \to \cP(\Sigma^\lo)$ an \emph{infinitary effect typing for~$M$} 
if it is the \emph{greatest solution} of the equation system $\{ \delta = S_\delta \mid \delta \in \Sig \}$.
\end{definition}

In the equation system in this definition, we treat method signatures as
variables ranging over $\cP(\Sigma^\lo)$.
We can then consider call expressions to define languages in $\cP(\Sigma^\lo)$
too, by interpreting the operations in them as language operations. 

We will prove the correctness of the non-terminating effects by induction on an approximation to the greatest solution $\eta$, but begin with an example for the definition.

\begin{example}[Finite and Infinite Effects]
\label{example:typing}
A possible class table for the program from Example~\ref{example:code} is shown below. 
It uses the regions $r_1=\Createdat{\ell_1}$, $r_2=\Createdat{\ell_2}$ and $r_3=\Createdat{\ell_3}$.
In the table we use the abbreviation $\delta_r = (\Node,r,\ttlast,())$.
\begin{lstlisting}
class Node$@\Null$
  next : $\{\Null\}$
  last() : $\emptyset \calls \emptyset$ 
\end{lstlisting}
\begin{lstlisting}
class Node$@r_1$
  next : $\{\Null\}$
  last() : $r_1 \andrel \{a\} \calls \{a\}\cdot\delta_\Null$
\end{lstlisting}
\begin{lstlisting}
class Node$@r_2$
  next : $\{\Null,\, r_1\}$
  last() : $r_1 \andrel \{aa\} \mid r_2 \andrel \{a\} \calls \{a\}\cdot\delta_{r_1}$
\end{lstlisting}
\begin{lstlisting}
class Node$@r_3$
  next : $\{\Null,\, r_3\}$
  last() : $r_3 \andrel a^+ \calls \{a\}\cdot\delta_{r_3}$
\end{lstlisting}
\begin{lstlisting}
class Test$@\Unknown$
  linear() : $r_1 \andrel \{aa\} \mid r_2 \andrel \{a\} \calls\delta_{r_2}$
  cyclic() : $r_3 \andrel a^+ \calls \delta_{r_3}$
\end{lstlisting}
The members of class \Node{} are given different types depending on the region. In region~$r_1$, the field \texttt{next} has region $\Null$. This means that \texttt{next} can only be $\itnull$. As a result, method \ttlast{} will always have~$a$ as its trace, thus its terminating effect is $\{a\}$. In region~$r_2$, the field \texttt{next} has region $\{\Null, r_1\}$. If method \ttlast{} returns \this{} then it has region $r_2$ and effect $\{a\}$; otherwise, it returns \texttt{next}.\ttlast{()} and thus has region $r_1$ and effect $\{aa\}$, where the second $a$ is emitted by \ttlast{} at $r_1$. This explains the annotation $r_1 \andrel \{aa\} \mid r_2 \andrel \{a\}$. In region $r_3$, we have a circularity, as \texttt{next} has also region~$r_3$.

We have the following equation system from the method typing:
\begin{align*}
\delta_\Null &= \emptyset
&
\delta_{r_3} &= \{a\} \cdot \delta_{r_3}
\\
\delta_{r_1}&= \{a\} \cdot \delta_{\Null}
&
(\Test,\Unknown,\linear,()) &= \delta_{r_2}
\\
\delta_{r_2} &= \{a\} \cdot \delta_{r_1}
&
(\Test,\Unknown,\cyclic,()) &= \delta_{r_3}
\end{align*}
The greatest solution $\eta$ maps $\delta_\Null$, $\delta_{r_1}$, $\delta_{r_2}$ and $(\Test,\Unknown$, $\linear,())$ to $\emptyset$, and $\delta_{r_3}$ and $(\Test,\Unknown,\cyclic,())$ to $a^\omega$, and is suitable as an infinitary effect typing.
\qed
\end{example}

\subsection{Soundness of the Type System}
\label{sec:soundness}

To formulate the soundness of the region type system w.r.t.\ the operational semantics, we give a formal interpretation of regions as a relation $(v,h) \vdash r$, read as ``the value $v$ at heap $h$ satisfies the property $r$'', generated from the following rules.
\[
\inferrule*
{ }
{(\itnull,h) \vdash \Null}
\qquad
\inferrule*
{ }
{(v,h) \vdash \Unknown}
\]
\[
\inferrule*
{h(l) = (C,G,\ell) \\
 C \in \Cls(\Createdat \ell)}
{(l,h) \vdash \Createdat \ell}
\]
It is lifted to a relation between stores and typing environments and one between heaps and field typings as follows:
\begin{itemize}
\item $(s,h) \vdash \Gamma$ iff $(s(x),h) \vdash r$ for all $(x\of r) \in \Gamma$.
\item $h \vdash F$ iff $(G(f),h) \vdash F(C,r,f)$ for all $l \in \dom(h)$ with $h(l)=(C,G,\_)$ and for all $r,f$ with $(C,r,f) \in \dom(F)$.
\end{itemize}
In words, $(s,h) \vdash \Gamma$ expresses that all the values in the store~$s$ satisfy the properties specified in the environment $\Gamma$, and $h \vdash F$ says that the fields of each object in the heap~$h$ satisfy the properties specified in the field typing $F$. We write $(s,h) \vdash \Gamma, F$ to denote the conjunction of $(s,h) \vdash \Gamma$ and $h \vdash F$.

\begin{theorem}[Soundness]
\label{thm:soundness:gen}
Let $P$ be a well-typed program w.r.t.\ a class table $(F,M)$. Let~$\eta$ be an infinitary effect typing for~$M$. For any\/ $\Gamma$, $e$, $T$, $S$, $s$, $h$, $v$, $h'$ and $w$ with
\[
\Gamma \vdash e : T \calls S
\quad \text{and} \quad
(s,h) \vdash \Gamma, F
\]
\begin{enumerate}
\item if\/ $(s,h) \vdash e \Downarrow v, h' \andrel w$, then $(s,h') \vdash \Gamma,F$ and $(v,h') \vdash r$ and $w \in U$ for some $(r \andrel U) \in T$;
\item if\/ $(s,h) \vdash e \Uparrow \andrel w$, then $w \in S(\eta)$, where $S(\eta) \subseteq \Sigma^\lo$ is obtained by substituting every occurrence of\/ $\delta$ in $S$ by $\eta(\delta)$.
\end{enumerate}
\end{theorem}
\begin{proof}[Proof of Theorem~\ref{thm:soundness:gen}(1)]
We carry out the proof by induction over the sum of the depth of the derivation of the typing judgment and the depth of the derivation of the operational semantics judgment. We refer the reader to a similar proof of \cite[Theorem~5]{EHZ:enforcing}.
\end{proof}

We focus on the soundness of the infinitary effect analysis. Recall that any expression $S \in \cP(\Sigma^*)\langle \Sig \rangle$ is of the form $\bigcup_{\delta \in \Delta} A_\delta \cdot \delta$ for some $\Delta \subseteq \Sig$. If we treat the signatures as variables then each expression gives a $\cP(\Sigma^\lo)$-valued function, \ie, for any assignment $\eta : \Sig \to \cP(\Sigma^\lo)$, we obtain a language $S(\eta) \subseteq \Sigma^\lo$ by substituting all $\delta$ in $S$ by $\eta(\delta)$. It is obvious that if $S \subseteq S'$ then $S(\eta) \subseteq S'(\eta)$ for any $\eta$.

For the proof of Theorem~\ref{thm:soundness:gen}(2), we use the following approximations of the greatest solution: Given a class table~$(F,M)$, let $\{ S_\delta \mid \delta \in \Sig \}$ be the call expressions encoded in~$M$. We define for each $n \in \N$ an assignment function $\eta_n : \Sig \to \cP(\Sigma^\lo)$ by
\[
\eta_0(\delta) = \Sigma^\lo
\qquad
\eta_{n+1}(\delta) = S_\delta(\eta_n).
\]
Let $\eta$ be the greatest solution of $\{ \delta = S_\delta \mid \delta \in \Sig \}$. Then we have $\eta(\delta) = \bigcap_{i \in \N} \eta_i(\delta)$ for every $\delta \in \Sig$. More generally, we have $\eta(\delta) = \bigcap_{i \in N} \eta_i(\delta)$ for any infinite subset $N \subseteq \N$.

We firstly show that the above approximations satisfy the soundness theorem in the following sense:
\begin{lemma}
\label{lem:soundness:approx}
Given a well-typed program w.r.t.\ $(F,M)$, let $\eta_n$ be the approximations to the greatest solution as defined above. For any\/ $n$, $\Gamma$, $e$, $T$, $S$, $s$, $h$ and $w$, if
\[
\Gamma \vdash e : T \calls S
\quad \text{and} \quad
(s,h) \vdash \Gamma,F
\quad \text{and} \quad
(s,h) \vdash e \Uparrow \andrel w
\]
then $w \in S(\eta_n)$.
\end{lemma}
\begin{proof}
By induction on $n$ and on the depth of the typing derivation.
When $n=0$, it is trivial because $\eta_0$ assigns $\Sigma^\lo$ to all variables.
Assume that the statement holds for $n$. To show that it holds for $n+1$, we perform an induction on the depth of the typing derivation.

If the last used rule of the typing derivation is \textsc{sub} or \textsc{weak}, then the induction hypothesis gives the desired result. Otherwise, we perform a case distinction over the last rule used in the derivation of the operational semantics judgment. The cases of $x$, $\tto(a)$, $\itnull$ and $\new{}{C}$ are ruled out because their execution always terminates. Here we check the following cases as examples:

Suppose the last used typing rule is \textsc{let}, \ie, the judgment is
\[
\textstyle
\Gamma \vdash \letin{x}{e_1}{e_2} : \bigcup^n_{i=1} U_i \cdot T_i \calls S \cup \bigcup^n_{i=1} U_i \cdot S_i.
\]
By rule inversion, we have $\Gamma \vdash e_1 : r_1 \andrel U_1 \mid \ldots \mid r_n \andrel U_n \calls S$ and $\Gamma,\, x\of r_i \vdash e_2 : T_i \calls S_i$ for all $i \in \{1,\ldots,n\}$.
Assume $(s,h) \vdash \Gamma,F$ and $(s,h) \vdash \letin{x}{e_1}{e_2} \Uparrow \andrel w$. By rule inversion, we have two possibilities:
\begin{itemize}
\item If $(s,h) \vdash e_1 \Downarrow v_1 , h_1 \andrel w_1$ and $(s[x\mapsto v_1],h_1) \vdash e_2 \Uparrow \andrel w_2$, then $w = w_1w_2$. By Theorem~\ref{thm:soundness:gen}(1), we have $(s,h_1) \vdash \Gamma,F$ and $(v_1,h_1) \vdash r_i$ and $w_1 \in U_i$ for some $i$. From the former we get $(s[x\mapsto v_1],h_1) \vdash (\Gamma,x\of r_i),F$. Then by induction hypothesis we know $w_2 \in S_i(\eta_{n+1})$, and thus $w = w_1w_2 \in U_i \cdot S_i(\eta_{n+1}) \subseteq (S \cup \bigcup^n_{i=1} U_i \cdot S_i)(\eta_{n+1})$.
\item If $(s,h) \vdash e_1 \Uparrow \andrel w$, then by induction hypothesis we have $w \in S(\eta_{n+1}) \subseteq (S \cup \bigcup^n_{i=1} U_i \cdot S_i)(\eta_{n+1})$.
\end{itemize}

Suppose the last used typing rule is \textsc{call}, \ie, the judgment is
\[
\textstyle
\Gamma,\, x \of r ,\, \bar{y} \of \bar{s} \vdash x^C.m(\bar{y}) : T \calls (C,r,m,\bar{s})
\]
where $(T,\_) = M(C,r,m,\bar{s})$. Assume $(s,h) \vdash (\Gamma, x \of r, \bar{y} \of \bar{s}), F$ and $(s,h) \vdash x.m(\bar{y}) \Uparrow \andrel w$. The goal is to show $w \in \eta_{n+1}(C,r,m,\bar{s})$.
Let $s(x) = l$ and $h(l)=(D,\_,\_)$ and $\mtable(D,m)=(\bar{z},e)$ and $s' = [\this \mapsto l]\cup[z_i \mapsto s(y_i)]_{i \in \{1,\ldots,|\bar{z}|\}}$. We have $(s',h) \vdash e \Uparrow \andrel w$ by rule inversion. From the assumption $(s,h) \vdash (\Gamma, x \of r, \bar{y} \of \bar{s}), F$ we have $(l,h) \vdash r$ which implies $D \in \Cls(r)$. Because the program is well-typed, we have $\Gamma' \vdash e : T' \calls S'$ where $\Gamma' = \this \of r, \bar{z} \of \bar{s}$ and $(T',S') = M(D,r,m,\bar{s})$. From $(s,h) \vdash (\Gamma, x \of r, \bar{y} \of \bar{s}), F$ we have also $(s',h) \vdash \Gamma',F$. Then the induction hypothesis gives us $w \in S'(\eta_n)$. Because $M$ is well-formed and $D \preceq C$, we have $S' \subseteq S_{C,r,m,\bar{s}}$ and thus $w \in S_{C,r,m,\bar{s}}(\eta_n) = \eta_{n+1}(C,r,m,\bar{s})$ by the definition of the approximations.

We leave the remaining cases to the reader.
\end{proof}

\begin{lemma}
For any call expression $S$ and infinite set $N \subseteq \N$, we have $\bigcap_{i \in N} S(\eta_i) \subseteq S(\eta)$.
\end{lemma}
%\begin{proof}
%By induction on the number of signatures in $S$.
%\end{proof}

The second half of the soundness theorem follows from the above lemmas:
\begin{proof}[Proof of Theorem~\ref{thm:soundness:gen}(2)]
If $\Gamma \vdash e : T \calls S$ and $(s,h) \vdash \Gamma,F$ and $(s,h) \vdash e \Uparrow \andrel w$, then we have $w \in \bigcap_{n \in \N} S(\eta_n) \subseteq S(\eta)$ by the above lemmas.
\end{proof}

Therefore, if a program is well-typed, then the effects encoded in the method typing are correct in the following sense:

\begin{corollary}[Soundness]
\label{thm:M:sound}
Suppose that a program is well-typed w.r.t.\ $(F,M)$ and that $\eta$ is an infinitary effect typing for~$M$.
Then, for any $(C,r,m,\bar{s}) \in \Sig$ with $M(C,r,m,\bar{s}) = (T,S)$ and $\eta(C,r,m,\bar{s}) = V$ and for any $x\of r$ and $\bar{y}\of\bar{s}$, we have
\begin{enumerate}
\item if $x.m(\bar{y})$ evaluates to a value $v$ and generates a trace $w$, then $v$ is in region $t$ and $w \in U$ for some $(t \andrel U) \in T$; and
\item if $x.m(\bar{y})$ diverges and generates a trace $w$, then $w \in V$.
\end{enumerate}
\end{corollary}

As a result, to verify whether a program adheres to a guideline, it suffices to check if the effects encoded in the typings $M$ and $\eta$ are allowed by the guideline.

\section{B\"uchi Effect Type System}
\label{sec:buechi}

The effect type system introduced in Section~\ref{sec:type:system:sets} is not yet suitable for practical use. Because effects are given as arbitrary languages, we have no direct way to compute greatest fixed points of equation systems over $\cP(\Sigma^\lo)$. In this section, we introduce a system (Section~\ref{sec:buechi:system}) where effect annotations are taken from a \emph{finite} structure to represent the languages over~$\Sigma$ (Section~\ref{sec:buechi:abstraction}) using the methods from abstract interpretation~\cite{cousot}. In Section~\ref{sec:type:inference}, we present a type inference algorithm that infers a class table for a given program, provided a standard FJ typing for the program is given.

\subsection{B\"uchi Abstraction}
\label{sec:buechi:abstraction}

Instead of $\cP(\Sigma^*)$ and $\cP(\Sigma^\lo)$, we work with a finite structure that abstracts the languages over $\Sigma$ for effect annotations, over which we can compute fixed points.

\begin{definition}[B\"uchi algebra]
Let $\cM_*$, $\cM_\lo$ be join-semilattices, that is, a partial order together with a join operation. We use the symbols $\sqsubseteq$, $\sqcup$ to refer to their generic structure of join-semilattices. We call $(\cM_*,\cM_\lo)$ a \emph{B\"uchi algebra} if it is equipped with the following operations:
\begin{itemize}
\item a binary operation on $\cM_*$, written as $\cA \cdot \cB$ for $\cA,\cB \in \cM_*$,
\item a mapping $\cM_* \times \cM_\lo \to \cM_\lo$, also written as $\cA \cdot \cV$ for $\cA \in \cM_*$ and $\cV \in \cM_\lo$, and
\item an operation $(-)^\omega : \cM_* \to \cM_\lo$,
\end{itemize}
such that all the operations are monotone
and the binary operations are associative, \ie, $\cA \cdot (\cB \cdot \cU) = (\cA \cdot \cB) \cdot \cU$ for any $\cA,\cB \in \cM_*$ and $\cU \in \cM_* \cup \cM_\lo$.
\end{definition}

For instance, $\cP(\Sigma^*)$ and $\cP(\Sigma^\lo)$ form a B\"uchi algebra. They are complete lattices. But note that meets are not required here because set intersections are not needed in the type system. We are looking for a \emph{finite} B\"uchi algebra $(\cM_*,\cM_\lo)$ that abstracts them.

\begin{remark}
The reader may have found our notion of B\"uchi algebra similar to the one of $\omega$-semigroup or Wilke algebra~\cite{CPP:automata:semigroups,PP:semigroups:automata}. One difference is that the carrier is equipped with a lattice structure in order to abstract the operations on languages. And the other is that there is no condition on the operation of $\omega$-iteration. We learned the concept of B\"uchi algebra from Martin Hofmann. In~\cite{10.5555/3329995.3330055} he and Chen define a B\"uchi algebra to be an $\omega$-semigroup whose carriers are complete lattices. Our specific definition requires weaker conditions as it was guided by the needs of this paper. Our main interest is the concept of B\"uchi abstraction which is defined below. We introduce B\"uchi algebras to capture the essential structure for it. With our definitions, a B\"uchi algebra may not necessarily be a Wilke algebra, but every B\"uchi abstraction is a Wilke algebra (see Proposition~\ref{prop:abs:properties}(5)).
\end{remark}

Recall that a \emph{Galois insertion} between partial orders $A,C$ consists of monotone maps $\alpha : C \to A$ and $\gamma : A \to C$ such that, for all $a \in A$ and $c \in C$,
\[
\alpha(\gamma(a)) = a
\qquad
\gamma(\alpha(c)) \sqsupseteq c.
\]
Intuitively, the function $\alpha$ furnishes for each concrete value in $C$ its \emph{abstraction}, whereas $\gamma$ \emph{concretizes} the abstract values in $A$. Moreover, the abstraction of $c$ has a concretization lying above $c$.

Now we are ready to introduce the finite structure needed for presenting the type system and the type inference algorithm.

\begin{definition}[B\"uchi abstraction]
A \emph{B\"uchi abstraction} is a \emph{finite} B\"uchi algebra $(\cM_*,\cM_\lo)$ equipped with Galois insertions $(\alpha_*,\gamma_*)$ between $\cM_*$ and $\cP(\Sigma^*)$ and $(\alpha_\lo,\gamma_\lo)$ between $\cM_\lo$ and $\cP(\Sigma^\lo)$ such that the abstraction functions $\alpha_*$ and $\alpha_\lo$ are compatible with the algebraic structure in the sense that
\[
\begin{aligned}
\alpha_*(A \cdot B) & = \alpha_*(A) \cdot \alpha_*(B) \\
\alpha_\lo(A \cdot V) & = \alpha_*(A) \cdot \alpha_\lo(V) \\
\alpha_\lo(A^\omega) & = \alpha_*(A)^\omega
\end{aligned}
\]
for all $A,B \subseteq \Sigma^*$ and $V \subseteq \Sigma^\lo$.
\end{definition}

One may have expected more structures on B\"uchi abstraction due to the analogy to $(\cP(\Sigma^*),\cP(\Sigma^\lo))$. Some structure can be constructed. For instance, the least elements of $\cM_*$ and $\cM_\lo$ are exactly the abstraction of the empty set.
We define the finite iteration $(-)^*$ on $\cM_*$ as the least fixed point $\cA^* = \lfp(\lambda X. \alpha_*(\{\varepsilon\}) \sqcup \cA \cdot X)$, and then have $\alpha_*(A^*) = (\alpha_*(A))^*$. Moreover, we have the following properties of the B\"uchi abstraction.

\begin{proposition}
\label{prop:abs:properties}
B\"uchi abstractions have the following properties:
\begin{enumerate}
\item The abstraction functions $\alpha_*$ and $\alpha_\lo$ preserve joins, \ie,
\[
\begin{aligned}
\alpha_*(A\cup B) &= \alpha_*(A) \sqcup \alpha_*(B) \\
\alpha_\lo(U\cup V) &= \alpha_\lo(U) \sqcup \alpha_\lo(V).
\end{aligned}
\]
\item The concretization functions $\gamma_*$ and $\gamma_\lo$ may not preserve joins or concatenations, but they satisfy
\[
\begin{aligned}
\gamma_*(\cA \sqcup \cB) & \supseteq \gamma_*(\cA) \cup \gamma_*(\cB) \\
\gamma_\lo(\cU \sqcup \cV) & \supseteq \gamma_\lo(\cU) \cup \gamma_\lo(\cV) \\
\gamma_*(\cA \cdot \cB) & \supseteq \gamma_*(\cA) \cdot \gamma_*(\cB) \\
\gamma_\lo(\cA \cdot \cV) & \supseteq \gamma_*(\cA) \cdot \gamma_\lo(\cV).
\end{aligned}
\]
\item The concatenation operations are distributive over joins, \ie,
\[
\begin{aligned}
(\cA \sqcup \cB ) \cdot \cV & = \cA \cdot \cV \sqcup \cB \cdot \cV \\
\cA \cdot (\cU \sqcup \cV) & = \cA \cdot \cU \sqcup \cA \cdot \cV.
\end{aligned}
\]
\item We have $\alpha_\lo(\gamma_*(\cA)^\omega) = \cA^\omega = \cA \cdot \cA^\omega$.
\item Every B\"uchi abstraction $(\cM_*,\cM_\lo)$ is a Wilke Algebra, \ie, for all $\cA,\cB \in \cM_*$ we have
\[
\begin{aligned}
\cA \cdot (\cB \cdot \cA)^\omega & = (\cA \cdot \cB)^\omega \\
(\cA^n)^\omega & = \cA^\omega \quad \text{for each } n>0 
\end{aligned}
\]
where $\cA^n$ is the product of $n$ copies of $\cA$, and every element of $\cM_\lo$ can be written as $\cA \cdot \cB^\omega$ for some $\cA,\cB \in \cM_*$.
\end{enumerate}
\end{proposition}

Note that, differing from finite iteration, here we do not define (or require) $\omega$-iteration to be the greatest fixed point. We have only $\cA^\omega \sqsubseteq \gfp(\lambda X.\cA \cdot X)$ in general, because $\cA^\omega$ is a fixed point of $\lambda X.\cA \cdot X$ by Proposition~\ref{prop:abs:properties}(4). In the following example, the $\omega$-iteration is strictly smaller than the greatest fixed point.

\begin{example}[Büchi Abstraction]
  \label{ex:abstr}
  Let the alphabet be $\Sigma = \{a\}$.
  We give a concrete example for a B\"uchi abstraction that distinguishes empty from non-empty words:
  \[
  \begin{aligned}
  \cM_* & = \{ \emptyset, \{\varepsilon\}, a^+, a^*\} \\
  \cM_\lo & = \{ \emptyset, \{\varepsilon\}, a^+, a^*, a^\omega, \{\varepsilon\} \cup a^\omega,  a^+ \cup a^\omega, a^* \cup a^\omega \}.
  \end{aligned}
  \]
  For both, $\sqcup$ and $\sqsubseteq$ are set union and inclusion. All other operations are given by the evident operations on languages.
  The abstraction function~$\alpha_*$ is determined by
  $\alpha_*(\{\varepsilon\}) = \{\varepsilon\}$ and
  $\alpha_*(X) = a^+$ for any $X\subseteq a^+$.
  All other values are determined by preservation of least element and joins.
  Similarly, $\alpha_\lo$ is determined by
  $\alpha_\lo(\{\varepsilon\}) = \{\varepsilon\}$,
  $\alpha_\lo(X) = a^+$ for any $X\subseteq a^+$ and
  $\alpha_\lo(a^\omega) = a^\omega$.
  The concretion functions $\gamma_*$ and $\gamma_\lo$ are both the identity.

  We note that the greatest solution of the equation $X = a^+ \cdot X$ over
  $\cM_\lo$ is $a^+ \cup a^\omega$ just by direct calculation.
  When calculated in $\cP(\Sigma^\lo)$, the greatest solution is $a^\omega$, which is the $\omega$-iteration of $a^+$ over $\cM_\lo$.
  \qed
\end{example}

As will become clear, a B\"uchi abstraction is sufficient for presenting the type system and the type inference algorithm. Hofmann and Chen~\cite{HC:buechi,HC:abstract} construct a B\"uchi abstraction from an \emph{extended B\"uchi automaton}~$\fA$ whose language is the union of the language of~$\fA$ when understood as a traditional nondeterministic finite automaton and the language of $\fA$ when understood as a traditional B\"uchi automaton. A programming guideline is simply a set of finite and infinite traces, and thus can be formalized by an extended B\"uchi automaton. Moreover, the B\"uchi abstraction of Hofmann and Chen is faithful with respect to acceptance by the automaton, which makes it suitable for representing effects in our type system. However, the development of the system (Section~\ref{sec:buechi:system}) and the type inference algorithm (Section~\ref{sec:type:inference}) are \emph{independent} of the construction of the B\"uchi abstraction. Readers should be able to access them without probing into the construction of Hofmann and Chen.

Similarly to the development in the previous section, we work with formal expressions $\cM_*\langle\Reg\rangle$ and $\cM_*\langle\Sig\rangle$ in the type system. We define the order $\sqsubseteq$ and the operations of concatenation $(\cdot)$ and join $(\sqcup)$ on expressions in the same way as in Definition~\ref{def:formal:expr}. To relate to the previous system, we extend the abstraction and concretization functions to expressions: Let $K$ be either $\Reg$ or $\Sig$. Given $\cT \in \cM_*\langle K \rangle$, we concretize it to $\gamma(\cT) \in \cP(\Sigma^*)\langle K \rangle$ by defining $\gamma(\cT)(x) = \gamma_*(\cT(x))$. Similarly, we abstract $T \in \cP(\Sigma^*)\langle K \rangle$ to $\alpha(T) \in \cM_*\langle K \rangle$. By the definition and properties of B\"uchi abstraction, we have:
\[
\begin{aligned}
\alpha(\gamma(\cT)) & = \cT &&&
\gamma(\alpha(T)) & \supseteq T \\
\alpha(T \cup T') & = \alpha(T) \sqcup \alpha(T') &&&
\gamma(\cT \sqcup \cT') & \supseteq \gamma(\cT) \cup \gamma(\cT') \\
\alpha(U \cdot T) & = \alpha_*(U) \cdot \alpha(T) &&&
\gamma(\cU \cdot \cT) & \supseteq \gamma_*(\cU) \cdot \gamma(\cT).
\end{aligned}
\]

\subsection{B\"uchi Effect Type System}
\label{sec:buechi:system}

In this section, we assume a programming guideline given by the language $\fA$ and a B\"uchi abstraction $(\cM_*,\cM_\lo)$ which is faithful w.r.t.\ $\fA$, \ie, $\gamma_\lo(\alpha_\lo(\fA)) = \fA$.

We now present an algorithmic type system with effects being elements of $\cM_*$ and $\cM_\lo$. The development is similar to the one in the previous section. We remove the subtyping rule and replace sets of traces by their abstractions in the type system. Thanks to the finite abstraction, we can compute greatest solutions to capture the possible non-terminating traces of a program. However, the abstraction may not preserve greatest fixed points in general. To make the analysis more precise, we concretize the abstract call expressions when defining infinitary effect typing. And we provide a simple algorithm to compute the infinitary effects that abstract the greatest fixed point of the concretized expressions.

In a \emph{B\"uchi class table} $(F,\Ma)$, the field typing $F : \Cls \times \Reg \times \Fld \rightharpoonup \Regs$ is the same as the one in the previous setting, while the method typing $\Ma : \Sig \to \cM_*\langle\Reg\rangle \times \cM_*\langle\Sig\rangle$ gives effect expressions with abstract sets of traces. We require $(F,\Ma)$ to be well-formed as in the previous setting.

A \emph{B\"uchi typing judgment} takes the form $\Gamma \vdashb e : \cT \calls \cS$ where $\cT : \cM_*\langle\Reg\rangle$ lists the possible regions and terminating effects of $e$ and $\cS: \cM_*\langle\Sig\rangle$ specifies the information about the calls that $e$ may make as in the previous system, except that effects are given by the B\"uchi abstraction. The typing rules are adapted correspondingly: Any singleton $\{a\}$ becomes $\alpha_*(\{a\})$ and the operations $\cup$ and $\cdot$ on languages are replaced by those from the B\"uchi algebra. The rules are listed in Figure~\ref{fig:buechi:typing}.
An FJ program is \emph{well-typed} w.r.t.\ $(F,\Ma)$ if for each $(C,r,m,\bar{s}) \in \Sig$ with $C \in \Cls(r)$ and $\Ma(C,r,m,\bar{s}) = (\cT,\cS)$ and $\mtable(C,m)=(\bar{x},e)$, the judgment $\this \of r,\, \bar{x} \of \bar{s} \vdashb e : \cT \calls \cS$ is derivable.

\begin{figure}[h]
\[
\inferrule*[left={sub}]
{\Gamma \vdashb e : \cT \calls \cS \quad
 \cT \sqsubseteq \cT' \quad
 \cS \sqsubseteq \cS'}
{\Gamma \vdashb e : \cT' \calls \cS'}
\]
\[
\inferrule*[left={weak}]
{\Gamma \vdashb e : (r \andrel \cU) \sqcup \cT \calls \cS }
{\Gamma \vdashb e : (\Unknown \andrel \cU) \sqcup \cT \calls \cS}
\]
\[
\inferrule*[left={prim}]
{\ }
{\Gamma \vdashb \tto(a) : \Null \andrel \alpha_*(\{a\}) \calls \emptyset}
\]
\[
\inferrule*[left={var}]
{\ }
{\Gamma,\, x\of r \vdashb x : r \andrel \alpha_*(\{\varepsilon\}) \calls \emptyset}
\]
\[
\inferrule*[left={if}]
{\Gamma,\, x\of r,\, y\of s \vdashb e_1 : \cT_1 \calls \cS_1 \quad
 \Gamma,\, x\of r,\, y\of s \vdashb e_2 : \cT_2 \calls \cS_2}
{\Gamma,\, x\of r,\, y\of s \vdashb \ttifthenelse{x}{y}{e_1}{e_2} : \cT_1 \sqcup \cT_2 \calls \cS_1 \sqcup \cS_2}
\]
\[
\inferrule*[left={else}]
{\mathit{disjoint}(r, s) \\
 \Gamma,\, x\of r,\, y\of s \vdashb e_2 : \cT_2 \calls \cS_2}
{\Gamma,\, x\of r,\, y\of s \vdashb \ttifthenelse{x}{y}{e_1}{e_2} : \cT_2 \calls \cS_2}
\]
\[
\inferrule*[left={let}]
{\Gamma \vdashb e_1 : r_1 \andrel \cU_1 \mid \ldots \mid r_n \andrel \cU_n \calls \cS\\
 \Gamma,\, x\of r_i \vdashb e_2 : \cT_i \calls \cS_i}
{\textstyle
 \Gamma \vdashb \letin{x}{e_1}{e_2} : \bigsqcup^n_{i=1} \cU_i \cdot \cT_i \calls \cS \sqcup \bigsqcup^n_{i=1} \cU_i \cdot \cS_i}
\]
\[
\inferrule*[left={null}]
{\ }
{\Gamma  \vdashb \ttnull : \Null \andrel \alpha_*(\{\varepsilon\}) \calls \emptyset}
\]
\[
\inferrule*[left={new}]
{\ }
{\Gamma \vdashb \new{\ell}{C} : \Createdat{\ell} \andrel \alpha_*(\{\varepsilon\}) \calls \emptyset}
\]
\[
\inferrule*[left={cast}]
{\Gamma \vdashb e : \cT \calls \cS}
{\Gamma \vdashb \cast{D}{e} : \cT \calls \cS}
\]
\[
\inferrule*[left={get}]
{\textstyle \cT = \bigsqcup \{ s \andrel \alpha_*(\{\varepsilon\}) \mid s \in F(C,r,f) \}}
{\Gamma,\, x\of r \vdashb x^C.f : \cT \calls \emptyset}
\]
\[
\inferrule*[left={set}]
{s \in F(C,r,f)}
{\Gamma,\, x\of r,\, y\of s \vdashb x^C.f := y : s \andrel \alpha_*(\{\varepsilon\}) \calls \emptyset}
\]
\[
\inferrule*[left={call}]
{(\cT,\_) = \Ma(C,r,m,\bar{s})}
{\Gamma,\, x\of r ,\, \bar{y} \of \bar{s} \vdashb x^C.m(\bar{y}) : \cT \calls (C,r,m,\bar{s})}
\]
\caption{Region-based B\"uchi effect type system}
\label{fig:buechi:typing}
\end{figure}

Given a B\"uchi method typing $\Ma$, we obtain a method typing $\Ma^\gamma$ with concretized effects by defining for $\delta \in \dom(\Ma)$ with $\Ma(\delta) = (\cT,\cS)$:
\[
\Ma^\gamma (\delta) = (\gamma(\cT),\gamma(\cS)).
\]
We relate the B\"uchi type system to the type system introduced in Section~\ref{sec:type:system:sets} as follows.

\begin{lemma}
\label{lm:derivable}
If\/ $\Gamma \vdashb e : \cT \calls \cS$ is derivable w.r.t.\ the B\"uchi class table $(F,\Ma)$, then so is $\Gamma \vdash e : \gamma(\cT) \calls \gamma(\cS)$ w.r.t.\ the corresponding class table $(F,\Ma^\gamma)$.
\end{lemma}

The following is a direct consequence of Lemma~\ref{lm:derivable}.

\begin{theorem}
\label{thm:buechi:soundness}
If a program~$P$ is well-typed w.r.t.\ a B\"uchi class table $(F,\Ma)$, then it is well-typed w.r.t.\ the corresponding class table $(F,\Ma^\gamma)$.
\end{theorem}

Moreover, we can compute the infinitary effect typing over the B\"uchi abstraction.

\begin{theorem}
\label{thm:buechi:infinitary}
Given a B\"uchi method typing $\Ma$, we can compute directly an assignment $\eta_\fA : \Sig \to \cM_\lo$ that abstracts the infinitary effect typing $\eta_\gamma : \Sig \to \cP(\Sigma^\lo)$ for $\Ma^\gamma$, \ie, for all $\delta \in \Sig$ we have $\eta_\fA(\delta) = \alpha_\lo(\eta_\gamma(\delta))$ where $\eta_\gamma$ is the greatest solution of $\{ \delta = \gamma(\cS_\delta) \mid \delta \in \Sig \}$ with $\cS_\delta$ the call expression in $\Ma(\delta)$. We call such an assignment $\eta_\fA$ a \emph{B\"uchi infinitary effect typing}.
\end{theorem}
\begin{proof}
Let $\cS_\delta$ be the call expressions in $\Ma(\delta)$. We solve the equation system $\{ \delta = \cS_\delta \mid \delta \in \Sig \}$ as follows: Recall that each expression has the form $\bigsqcup_{\delta \in \Delta} \cA_\delta \cdot \delta$ for some $\Delta \subseteq \Sig$. Viewing as an $\cM_\lo$-valued function, each $\cS_\delta$ can be rewritten to an equivalent form $\cA_\delta \cdot \delta \sqcup \mathcal{F}_\delta$ where the function $\mathcal{F}_\delta$ does not contain $\delta$, by the fact that the join operation is commutative. We repeat the following steps for each $\delta \in \Sig$:
\begin{enumerate}
\item Rewrite the right-hand side of the equation for $\delta$ to the form $\cA_\delta \cdot \delta \sqcup \mathcal{F}_\delta$.
\item Set $\delta = \cA_\delta^* \cdot \mathcal{F}_\delta \sqcup \cA^\omega_\delta$.
\item Substitute $\delta$ by $\cA_\delta^* \cdot \mathcal{F}_\delta \sqcup \cA^\omega_\delta$ on the right-hand side of all the other equations.
\end{enumerate}
One variable is eliminated on the right-hand side of the equations after each iteration. In the end we solve the system and obtain an assignment function $\eta_\fA : \Sig \to \cM_\lo$.

It abstracts the greatest solution $\eta_\gamma$ of $\{ \delta = \gamma(\cS_\delta) \mid \delta \in \Sig \}$ where the equations are concretized. For instance, consider a singleton equation system $\{ X = \cA \cdot X \sqcup \cB \}$. Its solution is
\[
\eta_\fA(X) = \cA^* \cdot \cB \sqcup \cA^\omega
\]
according to the above algorithm. The greatest solution of its concretization $\{ X = \gamma_*(\cA) \cdot X \cup \gamma_\lo(\cB) \}$ is $\eta_\gamma(X) = \gamma_*(\cA)^* \cdot \gamma_\lo(\cB) \cup \gamma_*(\cA)^\omega$. It is not hard to prove $\eta_\fA(X) = \alpha_\lo(\eta_\gamma(X))$ using Proposition~\ref{prop:abs:properties}. One can easily generalize the proof to systems with multiple equations. We refer the reader to \eg~\cite{fritz:fixed:point} for the characterization of simultaneous fixed points.
\end{proof}

\begin{remark}
One may expect a B\"uchi infinitary effect typing to be the greatest solution of $\{ \delta = \cS_\delta \mid \delta \in \Sig \}$ without concretization. However, the B\"uchi abstraction may not preserve greatest fixed points, which may lead to a loss in precision. Consider the example of a single method~$\texttt{f}$ with body $\tto(a); \this.\texttt{f}()$. The equation system is $\{ \delta_{\mathtt{f}} = \alpha_*(\{ a\}) \cdot \delta_{\mathtt{f}} \}$. If we solve it in the B\"uchi abstraction from Example~\ref{ex:abstr}, then we get the solution $\eta(\delta_{\mathtt{f}}) = a^+ \cup a^\omega$, as noted there. However, if we concretize this equation and solve over sets, then we get the solution $\eta_\gamma(\delta_{\mathtt{f}}) = a^\omega$. Therefore, we choose a better approximation of the infinite traces and formulate the B\"uchi infinitary effect typing as the abstraction of the greatest fixed point of the concretized call expressions.
\end{remark}

The soundness of the B\"uchi type system follows from the soundness of the previous type system (Theorem~\ref{thm:soundness:gen})
by the properties of abstraction and concretion.

\begin{corollary}[Soundness]
Suppose that a program is well-typed w.r.t.\ $(F,\Ma)$ and that $\eta_\fA$ is a Büchi infinitary effect typing for~$\Ma$.
Then, for any $(C,r,m,\bar{s}) \in \Sig$ with $\Ma(C,r,m,\bar{s}) = (\cT,\cS)$ and $\eta_\fA(C,r,m,\bar{s}) = \cV$, and for any $x\of r$ and $\bar{y}\of \bar{s}$, we have
\begin{enumerate}
\item if $x.m(\bar{y})$ evaluates to a value $v$ and generates a trace $w$, then $v$ is in region~$t$ and $w \in \gamma_*(\cU)$ for some $(t \andrel \cU) \in \cT$;
\item if $x.m(\bar{y})$ diverges and generates a trace $w$, then $w \in \gamma_\lo(\cV)$.
\end{enumerate}
\end{corollary}

To certify that all the traces of a given method are allowed by the language~$\fA$, we only need to check if all its effects stored in the typings $\Ma$ and $\eta_\fA$ are below the abstraction of~$\fA$. 

Let us spell out directly what this theorem means if we use the B\"uchi abstraction constructed from an
extended B\"uchi automaton using the construction of Hofmann and Chen~\cite{HC:buechi}.
In this case, $\fA$ is the language of the automaton. The Büchi abstraction of the automaton is
such that the automaton accepts a language $V\subseteq\Sigma^\lo$
if and only if $\alpha_\lo(V)$ is below the abstraction of~$\fA$. This follows from
\cite[Lemma 3(c)]{HC:buechi}.
Since we have defined a Büchi infinitary effect typing~$\eta_\fA$ as the abstraction of an 
infinitary effect typing~$\eta_\gamma$, this means that checking that $\eta_\fA(\delta)$ is below the abstraction of~$\fA$
amounts to checking that the extended Büchi automaton accepts all words in~$\eta_\gamma(\delta)$.

Come back to the example of a server given in the Introduction. Let $\fA$ be the
guideline consisting of traces where each event \textit{access} is immediately
preceded by \textit{authcheck}. It can be represented by a B\"uchi automaton.
Based on the automaton, we employ the construction of Hofmann and
Chen~\cite{HC:abstract} to get a B\"uchi abstraction. Checking whether the
non-terminating effect of \texttt{serve} is below~$\fA$ amounts to checking
whether it is accepted by the Büchi automaton. In this example, this will be the
case and \texttt{serve} is verified to adhere to the guideline.

\subsection{Type Inference Algorithm}
\label{sec:type:inference}

Given a program $P$ that is well-typed w.r.t.\ a given standard FJ class table, we extend the inference algorithm of~\cite[Appendix~F]{EHZ:enforcing} to construct a B\"uchi class table $(F,M_\fA)$ with respect to which $P$ is well-typed. The pseudo-code of our algorithm $\alg$ and its auxiliary procedures \ls{init} and \ls{typeff} is given in Figure~\ref{fig:alg}.

\begin{figure}[h]
\phantom{abc}
\begin{lstlisting}[language=myML]
proc $\alg(P)$
  $(F, \Ma)$ := init$(P)$
  do
    $(F', M')$ := $(F, \Ma)$ 
    foreach $(C,r,m,\bar{s}) \in \dom(\Ma)$ with $C \in \Cls(r)$
                                   and $(C,m)\in\dom(\mtable)$
      $(\bar{x},e)$ := $\mtable(C,m)$
      $\Gamma$ := $[\this \mapsto r]\cup[ x_i \mapsto s_i]_{i \in {\{1,\ldots,|\bar{x}|\}}}$
      $(\cT,\cS)$ := typeff$(\Gamma, e)$
      $\Ma(C,r,m,\bar{s})$ := $\Ma(C,r,m,\bar{s}) \sqcup (\cT,\cS)$
      $(F, \Ma)$ := checkClassTable$(F, \Ma)$
  until $(F',M') = (F, \Ma)$
  return $(F, \Ma)$

proc init$(P)$
  $\Cls(\Unknown)$ := $\Cls$
  $\Cls(\Null)$ := $\NullType$
  foreach $(\new{\ell}{C}) \in P$
    $\Cls(\Createdat{\ell})$ := $\Cls(\Createdat{\ell}) \cup \{ C \}$
  foreach $C \in \Cls$, $r \in \Reg$
    foreach $f \in \fields(C)$
      $F(C,r,f)$ := $\{ \Null \}$
    foreach $m \in \methods(C)$, $\bar{s} \in \Reg^{\mathrm{ar}(m)}$
      $\Ma(C,r,m,\bar{s})$ := $(\emptyset,\emptyset)$
  return $(F,\Ma)$

proc typeff$(\Gamma,e)$
  match $e$ with
  $\mid \tto(a)$ -> return $((\Null \andrel \alpha_*(\{a\})), \emptyset)$
  $\mid x$ -> return $((\Gamma(x) \andrel \alpha_*(\{\varepsilon\})), \emptyset)$
  $\mid \ttifthenelse{x}{y}{e_1}{e_2}$ ->
     if $\disjoint{\Gamma(x)}{\Gamma(y)}$ then return typeff$(\Gamma,e_2)$
     else return typeff$(\Gamma,e_1)$ $\sqcup$ typeff$(\Gamma,e_2)$
  $\mid \letin{x}{e_1}{e_2}$ ->
     $((r_1 \andrel \cU_1 \mid \ldots \mid r_n \andrel \cU_n),\cS)$ := typeff$(\Gamma,e_1)$
     foreach $i \in \{1,\ldots,n\}$
        $(\cT_i,\cS_i)$ := typeff$(\Gamma[x \mapsto r_i],e_2)$
     return $(\bigsqcup^n_{i=1} \cU_i \cdot \cT_i, \cS \sqcup \bigsqcup^n_{i=1} \cU_i \cdot \cS_i)$
  $\mid \ttnull$ -> return $((\Null \andrel \alpha_*(\{\varepsilon\})), \emptyset)$
  $\mid \new{\ell}{C}$ -> return $((\Createdat{\ell} \andrel \alpha_*(\{\varepsilon\})), \emptyset)$
  $\mid \cast{D}{e}$ -> return typeff$(\Gamma,e)$
  $\mid x^C.f$ -> return $(\bigsqcup\{ s \andrel \alpha_*(\{\varepsilon\}) | s \in F(C,\Gamma(x),f) \}, \emptyset)$
  $\mid x^C.f := y$ ->
     $F(C,\Gamma(x),f)$ := $F(C,\Gamma(x),f) \cup \{ \Gamma(y) \}$
     return $((\Gamma(y) \andrel \alpha_*(\{\varepsilon\})), \emptyset)$
  $\mid x^C.m(\bar{y})$ ->
     $(\cT,\_)$ := $\Ma(C,\Gamma(x),m,\Gamma(\bar{y}))$
     return $(\cT, (C,\Gamma(x),m,\Gamma(\bar{y})))$
\end{lstlisting}
\caption{Pseudo-code of the type inference algorithm}
\label{fig:alg}
\end{figure}

Specifically, $\alg$ starts by initializing $(F,\Ma)$ using \ls{init}: All entries in $F$ are set to be $\{\Null\}$, and those in $\Ma$ are set to be the pair $(\emptyset,\emptyset)$ of empty expressions. In words, we start with the assumption that all fields are $\itnull$ and all methods have no effects. Note that \ls{init} computes also the set $\Cls(r)$ of related classes for each region~$r$. In the definition of \ls{init}, we write $\mathrm{ar}(m)$ to denote the number of arguments of the method~$m \in \methods(C)$.

After initializing $(F,\Ma)$, $\alg$ infers the types and effects of each method body of the program, ``weakens'' the corresponding entry in~$\Ma$ if it differs from the inference result, and repeats this until no further updates of $(F,\Ma)$ are possible. The inference of types and effects is performed via the \typeff{} procedure that will be explained below. During each iteration, if the typings $F$ and $\Ma$ are updated, the procedure \ls{checkClassTable} ensures that they are well-formed in the sense of Definition~\ref{def:class:table}. For instance, if $C \preceq D$ and $F(C,r,f) \not= F(D,r,f)$ for some region~$r$ and field~$f$, then both entries are set to $F(C,r,f) \cup F(D,r,f)$ by \ls{checkClassTable}. Because there are only finitely many regions and abstract effects, $\alg$ will reach the fixed point after some iterations. The resulting table $(F,\Ma)$ encodes the terminating effects and the call expressions.

The \typeff{} procedure maps an environment $\Gamma : \Var \rightharpoonup \Reg$ and a term $e \in \Expr$ to formal expressions $\cT \in \cM_*\langle\Reg\rangle$ and $\cS \in \cM_*\langle\Sig\rangle$, if $e$ is typable in $\Gamma$, such that $\Gamma \vdashb e : \cT \calls~\cS$ is derivable. It can be viewed as a function representation of the typing rules listed in Figure~\ref{fig:buechi:typing}. For example, when $e$ is an if-expression $\ttifthenelse{x}{y}{e_1}{e_2}$, the procedure looks at the regions of $x$ and $y$; if they are disjoint, then it returns \typeff$(\Gamma,e_2)$; otherwise, it returns the joint of \typeff$(\Gamma,e_1)$ and \typeff$(\Gamma,e_2)$. This captures the idea of the typing rules \textsc{if} and \textsc{else}. Note that \typeff{} implicitly takes the class table $(F,\Ma)$ as input. For instance, when $e$ is a call expression $x^C.m(\bar{y})$, \typeff{} looks up its effects and types in the method typing $\Ma$. Moreover, \typeff{} may update the field typing~$F$ when~$e$ is a field-update expression $x^C.f := y$. The \textsc{set} rule requires that the region of $y$ must be in those allowed for the field $f$ by $F$. If this condition is not satisfied, then \typeff{} updates $F$ by joining the offending entry with the region of $y$. Note that the \textsc{weak} rule does not have to be applied, but it can be used to trade precision of analysis for efficiency of the inference algorithm. See the discussion in Section~\ref{sec:conclusion}. 

\begin{theorem}
Let $P$ be a program that is well-typed w.r.t.\ some standard FJ class table. Then the above algorithm gives a B\"uchi class table $(F,\Ma)$ with respect to which $P$ is well-typed.
\end{theorem}

The typing $\Ma$ encodes the terminating effects of the program. We employ the algorithm in the proof of Theorem~\ref{thm:buechi:infinitary} to compute also a typing $\eta_\fA$ of nonterminating effects.

\section{Experimental Evaluation}
\label{sec:implementation}

We have a prototype implementation\footnote{The implementation is available at \url{https://github.com/cj-xu/GuideForceJava}.} of the type inference algorithm (Section~\ref{sec:type:inference}) using the Soot framework~\cite{Soot}. In the implementation, we extend the approach explained in this paper to support exception handling (see Appendix~\ref{sec:exception} for the details of the extension). But other language features such as reflection and concurrency are not covered yet. Our tool takes a Java bytecode program and a programming guideline as inputs. The guideline is represented as an automaton. Using Hofmann and Chen's construction~\cite{HC:abstract}, a B\"uchi abstraction for representing effects is generated from the automaton. The guideline includes a configuration specifying the default effects of intrinsic methods. For instance, the \texttt{verifyAuthorization} method in the \texttt{serve} example would be specified to have an effect representing $authcheck$. From the standard Java typing provided by the Java bytecode, our tool infers the effect information of the program and then verifies if it is accepted by the guideline automaton. If the inferred effect is not acceptable, the tool tries to report a counterexample, \ie, searching for an execution path of the program that violates the guideline. Such an execution path can help programers to identify the bug(s) of the program.

We tested and evaluated our implementation with the Securibench Micro benchmark\footnote{The website of Securibench Micro: \url{http://too4words.github.io/securibench-micro/}} as well as a number of additional examples. The benchmark provides 122 test cases in 12 categories. Each benchmark program implements a small self-contained servlet. We ran our tool to verify if they adhere to the guideline that only untainted commands can be executed. In addition to the examples of the paper, we implemented 19 programs with typical nonterminating behaviors and 10 programs with exception handling for testing the infinitary and exceptional effect analysis. For these programs, we tested if our tool computes the correct effects. We ran the experiment on a machine with a 3,5 GHz Dual-Core Intel Core i7 processor. The analysis in each test takes around 200ms on average, and the slowest one takes almost a second.

The results of the experiment are given in Table~\ref{table:experiment}, where the results of the 12 categories of Securibench Micro are listed first, followed by the results of our three groups of test cases. A cell $n/m$ of the ``Pass'' column means that there are totally $m$ test cases and $n$ of them run as expected. When a test fails, a reference to the reason is given in the ``Comments'' column. Here are the reasons:
\begin{enumerate}
\item \label{item:array} We assume that if a tainted element is added into an array then the array is tainted (\ie, all its elements are tainted), which can cause false positives.
\item \label{item:path} Our analysis is not fully path-sensitive.
\item \label{item:field:update} Our field updates are conservatively treated as weak updates. For example, if a field $x.f$ has a region type $\{T\}$ indicating that it is tainted, and later it is updated to a untainted value, then the type of $x.f$ becomes $\{T,U\}$ where $U$ is the region for untainted values. This sometimes leads to false positives in taint analysis.
\item \label{item:clinit} Static initialization is not supported.
\item \label{item:reflection} Reflection is not supported.
\item \label{item:syn} Concurrent features are not supported.
\end{enumerate}
The ``CE Report'' column contains the results of searching counterexamples. A cell $i/j$ of it expresses that there are $j$ programs violating the guideline, and for $i$ of them an execution path that leads to the violation is found. 

\begin{table}[h]
\centering
\begin{tabular}{|c|c|c|c|} 
\hline
\textbf{Test Category} & \textbf{Pass} & \textbf{CE Report} & \textbf{Comments} \\ 
\hline\hline
aliasing & 6/6 & 5/5 &  \\ 
\hline
arrays & 9/10 & 9/9 & (\ref{item:array}) \\
\hline
basic & 42/42 & 42/42 &  \\
\hline
collections & 14/14 & 13/13 &  \\
\hline
data structures & 6/6 & 4/5 &  \\
\hline
factories & 3/3 & 3/3 &  \\
\hline
inter & 12/14 & 12/12 & (\ref{item:clinit}) \\
\hline
pred & 6/9 & 5/5 & (\ref{item:path}) \\
\hline
reflection & 0/4 & 0/0 & (\ref{item:reflection}) \\
\hline
sanitizers & 6/6 & 3/3 &  \\
\hline
session & 3/3 & 3/3 &  \\
\hline
strong updates & 3/5 & 1/1 & (\ref{item:field:update}), (\ref{item:syn}) \\
\hline\hline
paper examples & 8/8 & 1/1 &  \\
\hline
infinitary & 19/19 & 0/0 &  \\
\hline
exceptions & 10/10 & 0/0 &  \\
\hline\hline
\textbf{total} & 145/157 & 101/102 & \\
\hline
\end{tabular}

\medskip
\caption{Results of the experiment}
\label{table:experiment}
\end{table}

\section{Previous and Related Work}
\label{sec:related:work}

We work along the line of typing Java programs with regions and effects~\cite{GHL:enforcing,BGH:pointer,EHZ:enforcing}. Instead of refining the standard Featherweight Java type system as in the previous work, we introduce a flow type system where class information from the standard type system can be omitted. Our approach to region typing is in the spirit of Microsoft's TypeScript~\cite{typescript} and Facebook's Flow~\cite{flow}. Moving to a flow type system has the benefit of simplifying the type system and making the choice of regions more flexible. Our system is similar to the session type and effect system for PCF studied in~\cite{OY:effects:sessions}. But ours can capture also effects of non-terminating programs.

The idea of a type system with B\"uchi effects for Java has been sketched in~\cite{HC:abstract,HC:buechi}. But its meta-theoretic properties like soundness and the correctness of type inference were not investigated. Here we fully develop this idea into to a new flow type system, and establish the meta-theory. Furthermore, we make the type system independent from the automata-theoretic constructions of~\cite{HC:abstract,HC:buechi} by capturing the necessary structure abstractly in terms of a B\"uchi abstraction. B\"uchi automata have been employed also in model checking and testing for infinite executions of programs~\cite{holzmann:spin,ADALID201461,ZT:test:buechi}. Our approach is entirely based on type systems and abstract interpretation.

Regions serve as the basis of various techniques for \eg\ memory management~\cite{10.1145/1029873.1029884,10.1006/inco.1996.2613}, pointer analysis~\cite{BGH:pointer,10.1007/978-3-319-10936-7_20}, and race detection~\cite{10.1007/978-3-642-03237-0_13}. They are closely interrelated with effects. Thus the use of regions makes our type-based approach of analysis context-sensitive. In higher-order model checking~\cite{10.1109/LICS.2009.29,10.1145/2487241.2487246,10.1109/LICS.2015.71,suzuki_et_al:LIPIcs:2017:7732}, intersection types are employed instead to achieve the context-sensitivity. 

Our approach is closely related to the one of Nanjo \emph{et al.}~\cite{10.1145/3209108.3209204}. They work with a type system where effects are represented by formulas of some fixpoint logic. Specifically, their formulas may contain least and greatest fixed points to specify the finite and infinite behavior of the program, which is highly similar to what our call expressions do (see Section~\ref{sec:formal:expression}). Their effects are value-dependent, because the formulas can contain program values. This makes their analysis path-sensitive. In our setting, effects can depend on regions which are a notion of abstracted value. In the computation of infinitary effects (Section~\ref{sec:infinitary}), we use an equation system with method signatures as variables. Notice that method signatures come with a region for each argument. Therefore, we have a restricted form of dependent temporal effect. To capture more precise dependencies, it would be possible to refine the notion of region, perhaps even to refinement types~\cite{10.1145/113445.113468}.

Skalka \emph{et al.}~\cite{10.1007/s10990-008-9032-6,10.1017/S0956796807006466} also work with effect type systems to verify trace-based safety properties. But they apply model-checking tools to trace effects inferred by the type systems for the verification. In their setting, trace effects are represented as label transition systems (LTSs). It seems that the equation systems that we use for approximating infinite traces can be considered as a form of LTS, if one orients the equations, that plays the same role as the LTSs in~\cite{10.1007/s10990-008-9032-6,10.1017/S0956796807006466}. However, in our case the variables in the equation system (\ie, method signatures) depend on regions, which improves the precision of analysis. The effect type system of~\cite{10.1007/s10990-008-9032-6} is an refinement of the FJ type system. Ours is simpler because it is separate from the FJ type system. 
%(also over~\cite{HC:abstract,HC:buechi}).

\section{Conclusion and Discussion}
\label{sec:conclusion}

This paper presents a framework to verify if a given Featherweight Java program adheres to a given programming guideline. Let $\fA$ be a language representing the guideline. Suppose a B\"uchi abstraction that is faithful w.r.t.\ $\fA$ is given, \eg\ the one of Hofmann and Chen~\cite{HC:buechi}. From the standard Java types of the fields and methods, we construct a B\"uchi class table $(F,\Ma)$ with respect to which the program is well-typed. The table $\Ma$ and its infinitary effect typing $\eta_\fA$ provide the terminating and non-terminating effects to each method signature. If the effects of the methods are below the abstraction of $\fA$, then any event trace generated by the program is allowed by $\fA$ and thus the program adheres to the guideline.

%In further work, we intend to work with richer domains of regions, \eg~with lattice structure, to capture more detailed information about objects. %, \eg~to express properties like $\Createdat{\ell} \land \mathtt{Tainted}$.

Comparing to the previous work~\cite{EHZ:enforcing,HC:abstract}, our approach enhances the precision of analysis by capturing the effect for each region where a value may locate. The cost of this improvement is that the type inference is less efficient.
For instance, to compute the effect of the expression $\letin{x}{e_1}{e_2}$, we need to infer $e_2$ as many times as the number of possible regions where the value of $e_1$ may locate, while in the systems of~\cite{EHZ:enforcing,HC:abstract} $e_2$ is inferred only once.
However, our framework is actually very flexible. If we apply the \textsc{weak} rule in every step of the typing derivation, then we essentially work with an effect type system without regions, because in the end only the region \Unknown{} remains in the result and all effects of a value are merged.
If we extend regions with a join operation $\Union{r_1}{r_2}$ and a corresponding lattice structure, and change the \textsc{weak} rule to allow merging $r_1 \andrel U_1 \mid r_2 \andrel U_2$ to $\Union{r_1}{r_2} \mid (U_1 \cup U_2)$, then we can recover the systems of~\cite{EHZ:enforcing,HC:abstract}.
In addition, we can reduce the size of the method table in an implementation by letting $M(C, \Union{r_1}{r_2}, m, \bar s) = M(C, r_1, m, \bar s) \sqcup  M(C, r_2, m, \bar s)$ and storing only the latter two entries.
Moreover, one may come up with certain strategy about when to apply the \textsc{weak} rule in the inference algorithm to balance efficiency and precision. It is our future work to investigate such directions.

\begin{acks}
We thank the anonymous reviewers for their valuable comments and helpful suggestions. We are also grateful to the Mathematically Structured Programming Group at the University of Strathclyde for fruitful discussions on this work in the MSP 101 seminar. This work is supported by the \grantsponsor{dfg}{German Research Foundation}{https://www.dfg.de/} (DFG) under research grant~\grantnum{dfg}{250888164} (GuideForce).
\end{acks}

% \balance
\bibliography{GFbib}

%%% -*-BibTeX-*-
%%% Do NOT edit. File created by BibTeX with style
%%% ACM-Reference-Format-Journals [18-Jan-2012].

\begin{thebibliography}{32}

%%% ====================================================================
%%% NOTE TO THE USER: you can override these defaults by providing
%%% customized versions of any of these macros before the \bibliography
%%% command.  Each of them MUST provide its own final punctuation,
%%% except for \shownote{}, \showDOI{}, and \showURL{}.  The latter two
%%% do not use final punctuation, in order to avoid confusing it with
%%% the Web address.
%%%
%%% To suppress output of a particular field, define its macro to expand
%%% to an empty string, or better, \unskip, like this:
%%%
%%% \newcommand{\showDOI}[1]{\unskip}   % LaTeX syntax
%%%
%%% \def \showDOI #1{\unskip}           % plain TeX syntax
%%%
%%% ====================================================================

\ifx \showCODEN    \undefined \def \showCODEN     #1{\unskip}     \fi
\ifx \showDOI      \undefined \def \showDOI       #1{#1}\fi
\ifx \showISBNx    \undefined \def \showISBNx     #1{\unskip}     \fi
\ifx \showISBNxiii \undefined \def \showISBNxiii  #1{\unskip}     \fi
\ifx \showISSN     \undefined \def \showISSN      #1{\unskip}     \fi
\ifx \showLCCN     \undefined \def \showLCCN      #1{\unskip}     \fi
\ifx \shownote     \undefined \def \shownote      #1{#1}          \fi
\ifx \showarticletitle \undefined \def \showarticletitle #1{#1}   \fi
\ifx \showURL      \undefined \def \showURL       {\relax}        \fi
% The following commands are used for tagged output and should be
% invisible to TeX
\providecommand\bibfield[2]{#2}
\providecommand\bibinfo[2]{#2}
\providecommand\natexlab[1]{#1}
\providecommand\showeprint[2][]{arXiv:#2}

\bibitem[\protect\citeauthoryear{Adalid, Salmerón, del Mar~Gallardo, and
  Merino}{Adalid et~al\mbox{.}}{2014}]%
        {ADALID201461}
\bibfield{author}{\bibinfo{person}{Damián Adalid}, \bibinfo{person}{Alberto
  Salmerón}, \bibinfo{person}{María del Mar~Gallardo}, {and}
  \bibinfo{person}{Pedro Merino}.} \bibinfo{year}{2014}\natexlab{}.
\newblock \showarticletitle{Using {SPIN} for automated debugging of infinite
  executions of {J}ava programs}.
\newblock \bibinfo{journal}{\emph{Journal of Systems and Software}}
  \bibinfo{volume}{90} (\bibinfo{year}{2014}), \bibinfo{pages}{61--75}.
\newblock
\urldef\tempurl%
\url{https://doi.org/10.1016/j.jss.2013.10.056}
\showDOI{\tempurl}


\bibitem[\protect\citeauthoryear{Beringer, Grabowski, and Hofmann}{Beringer
  et~al\mbox{.}}{2013}]%
        {BGH:pointer}
\bibfield{author}{\bibinfo{person}{Lennart Beringer}, \bibinfo{person}{Robert
  Grabowski}, {and} \bibinfo{person}{Martin Hofmann}.}
  \bibinfo{year}{2013}\natexlab{}.
\newblock \showarticletitle{Verifying pointer and string analyses with region
  type systems}.
\newblock \bibinfo{journal}{\emph{Computer Languages, Systems \& Structures}}
  \bibinfo{volume}{39}, \bibinfo{number}{2} (\bibinfo{year}{2013}),
  \bibinfo{pages}{49--65}.
\newblock
\urldef\tempurl%
\url{https://doi.org/10.1016/j.cl.2013.01.001}
\showDOI{\tempurl}


\bibitem[\protect\citeauthoryear{Carton, Perrin, and Pin}{Carton
  et~al\mbox{.}}{2008}]%
        {CPP:automata:semigroups}
\bibfield{author}{\bibinfo{person}{Olivier Carton}, \bibinfo{person}{Dominique
  Perrin}, {and} \bibinfo{person}{Jean-Eric Pin}.}
  \bibinfo{year}{2008}\natexlab{}.
\newblock \showarticletitle{Automata and semigroups recognizing infinite
  words}.
\newblock In \bibinfo{booktitle}{\emph{Logic and Automata: History and
  Perspectives}}, \bibfield{editor}{\bibinfo{person}{J\"org Flum},
  \bibinfo{person}{Erich Gr\"adel}, {and} \bibinfo{person}{Thomas Wilke}}
  (Eds.). \bibinfo{series}{Texts in Logic and Games}, Vol.~\bibinfo{volume}{2}.
  \bibinfo{publisher}{Amsterdam University Press}, \bibinfo{pages}{133--168}.
\newblock


\bibitem[\protect\citeauthoryear{Cherem and Rugina}{Cherem and Rugina}{2004}]%
        {10.1145/1029873.1029884}
\bibfield{author}{\bibinfo{person}{Sigmund Cherem} {and} \bibinfo{person}{Radu
  Rugina}.} \bibinfo{year}{2004}\natexlab{}.
\newblock \showarticletitle{Region Analysis and Transformation for Java
  Programs}. In \bibinfo{booktitle}{\emph{Proceedings of the 4th International
  Symposium on Memory Management}} \emph{(\bibinfo{series}{ISMM '04})}.
  \bibinfo{publisher}{Association for Computing Machinery},
  \bibinfo{address}{New York, NY, USA}, \bibinfo{pages}{85--96}.
\newblock
\urldef\tempurl%
\url{https://doi.org/10.1145/1029873.1029884}
\showDOI{\tempurl}


\bibitem[\protect\citeauthoryear{Cousot and Cousot}{Cousot and Cousot}{1977}]%
        {cousot}
\bibfield{author}{\bibinfo{person}{Patrick Cousot} {and}
  \bibinfo{person}{Radhia Cousot}.} \bibinfo{year}{1977}\natexlab{}.
\newblock \showarticletitle{Abstract Interpretation: A Unified Lattice Model
  for Static Analysis of Programs by Construction or Approximation of
  Fixpoints}. In \bibinfo{booktitle}{\emph{Principles of Programming Languages
  (POPL 1977)}}, \bibfield{editor}{\bibinfo{person}{Robert~M. Graham},
  \bibinfo{person}{Michael~A. Harrison}, {and} \bibinfo{person}{Ravi Sethi}}
  (Eds.). \bibinfo{pages}{238--252}.
\newblock
\urldef\tempurl%
\url{https://doi.org/10.1145/512950.512973}
\showDOI{\tempurl}


\bibitem[\protect\citeauthoryear{Erbatur, Hofmann, and Z\u{a}linescu}{Erbatur
  et~al\mbox{.}}{2017}]%
        {EHZ:enforcing}
\bibfield{author}{\bibinfo{person}{Serdar Erbatur}, \bibinfo{person}{Martin
  Hofmann}, {and} \bibinfo{person}{Eugen Z\u{a}linescu}.}
  \bibinfo{year}{2017}\natexlab{}.
\newblock \showarticletitle{Enforcing Programming Guidelines with Region Types
  and Effects}. In \bibinfo{booktitle}{\emph{Programming Languages and Systems
  (APLAS~2017)}} \emph{(\bibinfo{series}{Lecture Notes in Computer Science})},
  \bibfield{editor}{\bibinfo{person}{Bor-Yuh~Evan Chang}} (Ed.),
  Vol.~\bibinfo{volume}{10695}. \bibinfo{publisher}{Springer, Cham},
  \bibinfo{pages}{85--104}.
\newblock
\urldef\tempurl%
\url{https://doi.org/10.1007/978-3-319-71237-6_5}
\showDOI{\tempurl}


\bibitem[\protect\citeauthoryear{Facebook}{Facebook}{[n.d.]a}]%
        {flow}
\bibfield{author}{\bibinfo{person}{Facebook}.}
  \bibinfo{year}{[n.d.]}\natexlab{a}.
\newblock \bibinfo{title}{{F}low - {A} static type checker for {J}ava{S}cript}.
\newblock
\newblock
\urldef\tempurl%
\url{https://flow.org}
\showURL{%
Retrieved July 22, 2021 from \tempurl}


\bibitem[\protect\citeauthoryear{Facebook}{Facebook}{[n.d.]b}]%
        {infer}
\bibfield{author}{\bibinfo{person}{Facebook}.}
  \bibinfo{year}{[n.d.]}\natexlab{b}.
\newblock \bibinfo{title}{Infer}.
\newblock
\newblock
\urldef\tempurl%
\url{https://fbinfer.com/}
\showURL{%
Retrieved July 22, 2021 from \tempurl}


\bibitem[\protect\citeauthoryear{Freeman and Pfenning}{Freeman and
  Pfenning}{1991}]%
        {10.1145/113445.113468}
\bibfield{author}{\bibinfo{person}{Tim Freeman} {and} \bibinfo{person}{Frank
  Pfenning}.} \bibinfo{year}{1991}\natexlab{}.
\newblock \showarticletitle{Refinement Types for ML}. In
  \bibinfo{booktitle}{\emph{Proceedings of the ACM SIGPLAN 1991 Conference on
  Programming Language Design and Implementation}} \emph{(\bibinfo{series}{PLDI
  '91})}. \bibinfo{publisher}{Association for Computing Machinery},
  \bibinfo{address}{New York, NY, USA}, \bibinfo{pages}{268--277}.
\newblock
\urldef\tempurl%
\url{https://doi.org/10.1145/113445.113468}
\showDOI{\tempurl}


\bibitem[\protect\citeauthoryear{Fritz}{Fritz}{2002}]%
        {fritz:fixed:point}
\bibfield{author}{\bibinfo{person}{Carsten Fritz}.}
  \bibinfo{year}{2002}\natexlab{}.
\newblock \showarticletitle{Some Fixed Point Basics}.
\newblock In \bibinfo{booktitle}{\emph{Automata Logics, and Infinite Games}},
  \bibfield{editor}{\bibinfo{person}{Erich Gr\"adel}, \bibinfo{person}{Wolfgang
  Thomas}, {and} \bibinfo{person}{Thomas Wilke}} (Eds.).
  \bibinfo{series}{Lecture Notes in Computer Science},
  Vol.~\bibinfo{volume}{2500}. \bibinfo{publisher}{Springer, Berlin,
  Heidelberg}, \bibinfo{pages}{359--364}.
\newblock
\urldef\tempurl%
\url{https://doi.org/10.1007/3-540-36387-4_20}
\showDOI{\tempurl}


\bibitem[\protect\citeauthoryear{Gosling, Joy, Steele, Bracha, and
  Buckley}{Gosling et~al\mbox{.}}{2020}]%
        {JavaSpec}
\bibfield{author}{\bibinfo{person}{James Gosling}, \bibinfo{person}{Bill Joy},
  \bibinfo{person}{Guy Steele}, \bibinfo{person}{Gilad Bracha}, {and}
  \bibinfo{person}{Alex Buckley}.} \bibinfo{year}{2020}\natexlab{}.
\newblock \bibinfo{booktitle}{\emph{The Java Language Specification}
  (\bibinfo{edition}{{J}ava {SE} 14} ed.)}.
\newblock \bibinfo{publisher}{Oracle America, Inc.}
\newblock
\urldef\tempurl%
\url{http://docs.oracle.com/javase/specs/jls/se14/jls14.pdf}
\showURL{%
\tempurl}


\bibitem[\protect\citeauthoryear{Grabowski, Hofmann, and Li}{Grabowski
  et~al\mbox{.}}{2012}]%
        {GHL:enforcing}
\bibfield{author}{\bibinfo{person}{Robert Grabowski}, \bibinfo{person}{Martin
  Hofmann}, {and} \bibinfo{person}{Keqin Li}.} \bibinfo{year}{2012}\natexlab{}.
\newblock \showarticletitle{Type-Based Enforcement of Secure Programming
  Guidelines — Code Injection Prevention at {SAP}}. In
  \bibinfo{booktitle}{\emph{Formal Aspects of Security and Trust (FAST~2011)}}
  \emph{(\bibinfo{series}{Lecture Notes in Computer Science})},
  \bibfield{editor}{\bibinfo{person}{G.~Barthe}, \bibinfo{person}{A.~Datta},
  {and} \bibinfo{person}{S.~Etalle}} (Eds.), Vol.~\bibinfo{volume}{7140}.
  \bibinfo{publisher}{Springer, Berlin, Heidelberg}, \bibinfo{pages}{182--197}.
\newblock
\urldef\tempurl%
\url{https://doi.org/10.1007/978-3-642-29420-4_12}
\showDOI{\tempurl}


\bibitem[\protect\citeauthoryear{Hofmann and Chen}{Hofmann and Chen}{2014a}]%
        {HC:abstract}
\bibfield{author}{\bibinfo{person}{Martin Hofmann} {and} \bibinfo{person}{Wei
  Chen}.} \bibinfo{year}{2014}\natexlab{a}.
\newblock \showarticletitle{Abstract Interpretation from {B}\"{u}chi Automata}.
  In \bibinfo{booktitle}{\emph{Proceedings of the Joint Meeting of the
  Twenty-Third EACSL Annual Conference on Computer Science Logic (CSL) and the
  Twenty-Ninth Annual ACM/IEEE Symposium on Logic in Computer Science (LICS)}}
  \emph{(\bibinfo{series}{CSL-LICS '14})}. \bibinfo{publisher}{Association for
  Computing Machinery}, \bibinfo{pages}{51:1--51:10}.
\newblock
\urldef\tempurl%
\url{https://doi.org/10.1145/2603088.2603127}
\showDOI{\tempurl}


\bibitem[\protect\citeauthoryear{Hofmann and Chen}{Hofmann and Chen}{2014b}]%
        {HC:buechi}
\bibfield{author}{\bibinfo{person}{Martin Hofmann} {and} \bibinfo{person}{Wei
  Chen}.} \bibinfo{year}{2014}\natexlab{b}.
\newblock \bibinfo{title}{B\"{u}chi Types for Infinite Traces and Liveness}.
\newblock
\newblock
\newblock
\shownote{Technical report, {\tt
  \href{https://arxiv.org/abs/1401.5107}{arXiv:1401.5107} [cs.LO]}.}


\bibitem[\protect\citeauthoryear{Hofmann and Ledent}{Hofmann and
  Ledent}{2017}]%
        {10.5555/3329995.3330055}
\bibfield{author}{\bibinfo{person}{Martin Hofmann} {and}
  \bibinfo{person}{J\'{e}r\'{e}my Ledent}.} \bibinfo{year}{2017}\natexlab{}.
\newblock \showarticletitle{A Cartesian-Closed Category for Higher-Order Model
  Checking}. In \bibinfo{booktitle}{\emph{Proceedings of the 32nd Annual
  ACM/IEEE Symposium on Logic in Computer Science}}
  \emph{(\bibinfo{series}{LICS '17})}. \bibinfo{publisher}{IEEE Press}, Article
  \bibinfo{articleno}{60}, \bibinfo{numpages}{12}~pages.
\newblock
\urldef\tempurl%
\url{https://doi.org/10.1109/LICS.2017.8005120}
\showDOI{\tempurl}


\bibitem[\protect\citeauthoryear{Holzmann}{Holzmann}{2003}]%
        {holzmann:spin}
\bibfield{author}{\bibinfo{person}{Gerard~J Holzmann}.}
  \bibinfo{year}{2003}\natexlab{}.
\newblock \bibinfo{booktitle}{\emph{The {SPIN} Model Checker: Primer and
  Reference Manual}}.
\newblock \bibinfo{publisher}{Addison-Wesley Professional, Boston, MA}.
\newblock


\bibitem[\protect\citeauthoryear{Igarashi, Pierce, and Wadler}{Igarashi
  et~al\mbox{.}}{2001}]%
        {DBLP:journals/toplas/IgarashiPW01}
\bibfield{author}{\bibinfo{person}{Atsushi Igarashi},
  \bibinfo{person}{Benjamin~C. Pierce}, {and} \bibinfo{person}{Philip Wadler}.}
  \bibinfo{year}{2001}\natexlab{}.
\newblock \showarticletitle{Featherweight {J}ava: a minimal core calculus for
  {J}ava and {GJ}}.
\newblock \bibinfo{journal}{\emph{ACM Transactions on Programming Languages and
  Systems}} \bibinfo{volume}{23}, \bibinfo{number}{3} (\bibinfo{year}{2001}),
  \bibinfo{pages}{396--450}.
\newblock
\urldef\tempurl%
\url{https://doi.org/10.1145/503502.503505}
\showDOI{\tempurl}


\bibitem[\protect\citeauthoryear{Kobayashi}{Kobayashi}{2013}]%
        {10.1145/2487241.2487246}
\bibfield{author}{\bibinfo{person}{Naoki Kobayashi}.}
  \bibinfo{year}{2013}\natexlab{}.
\newblock \showarticletitle{Model Checking Higher-Order Programs}.
\newblock \bibinfo{journal}{\emph{Journal of the {ACM}}} \bibinfo{volume}{60},
  \bibinfo{number}{3} (\bibinfo{year}{2013}), \bibinfo{pages}{20:1--20:62}.
\newblock
\urldef\tempurl%
\url{https://doi.org/10.1145/2487241.2487246}
\showDOI{\tempurl}


\bibitem[\protect\citeauthoryear{Kobayashi and Li}{Kobayashi and Li}{2015}]%
        {10.1109/LICS.2015.71}
\bibfield{author}{\bibinfo{person}{Naoki Kobayashi} {and} \bibinfo{person}{Xin
  Li}.} \bibinfo{year}{2015}\natexlab{}.
\newblock \showarticletitle{Automata-Based Abstraction Refinement for {µHORS}
  Model Checking}. In \bibinfo{booktitle}{\emph{2015 30th Annual ACM/IEEE
  Symposium on Logic in Computer Science (LICS'15)}}. \bibinfo{publisher}{IEEE
  Computer Society}, \bibinfo{pages}{713--724}.
\newblock
\urldef\tempurl%
\url{https://doi.org/10.1109/LICS.2015.71}
\showDOI{\tempurl}


\bibitem[\protect\citeauthoryear{Kobayashi and Ong}{Kobayashi and Ong}{2009}]%
        {10.1109/LICS.2009.29}
\bibfield{author}{\bibinfo{person}{Naoki Kobayashi} {and}
  \bibinfo{person}{C.-H.~Luke Ong}.} \bibinfo{year}{2009}\natexlab{}.
\newblock \showarticletitle{A Type System Equivalent to the Modal Mu-Calculus
  Model Checking of Higher-Order Recursion Schemes}. In
  \bibinfo{booktitle}{\emph{2009 24th Annual IEEE Symposium on Logic In
  Computer Science (LICS'09)}}. \bibinfo{publisher}{IEEE Computer Society},
  \bibinfo{pages}{179--188}.
\newblock
\urldef\tempurl%
\url{https://doi.org/10.1109/LICS.2009.29}
\showDOI{\tempurl}


\bibitem[\protect\citeauthoryear{Microsoft}{Microsoft}{[n.d.]}]%
        {typescript}
\bibfield{author}{\bibinfo{person}{Microsoft}.}
  \bibinfo{year}{[n.d.]}\natexlab{}.
\newblock \bibinfo{title}{{T}ype{S}cript - {T}yped {J}ava{S}cript at any
  scale}.
\newblock
\newblock
\urldef\tempurl%
\url{https://www.typescriptlang.org/}
\showURL{%
Retrieved July 22, 2021 from \tempurl}


\bibitem[\protect\citeauthoryear{Nanjo, Unno, Koskinen, and Terauchi}{Nanjo
  et~al\mbox{.}}{2018}]%
        {10.1145/3209108.3209204}
\bibfield{author}{\bibinfo{person}{Yoji Nanjo}, \bibinfo{person}{Hiroshi Unno},
  \bibinfo{person}{Eric Koskinen}, {and} \bibinfo{person}{Tachio Terauchi}.}
  \bibinfo{year}{2018}\natexlab{}.
\newblock \showarticletitle{A Fixpoint Logic and Dependent Effects for Temporal
  Property Verification}. In \bibinfo{booktitle}{\emph{2018 33rd Annual
  ACM/IEEE Symposium on Logic in Computer Science (LICS'18)}}.
  \bibinfo{publisher}{Association for Computing Machinery},
  \bibinfo{pages}{759--768}.
\newblock
\urldef\tempurl%
\url{https://doi.org/10.1145/3209108.3209204}
\showDOI{\tempurl}


\bibitem[\protect\citeauthoryear{Orchard and Yoshida}{Orchard and
  Yoshida}{2016}]%
        {OY:effects:sessions}
\bibfield{author}{\bibinfo{person}{Dominic Orchard} {and}
  \bibinfo{person}{Nobuko Yoshida}.} \bibinfo{year}{2016}\natexlab{}.
\newblock \showarticletitle{Effects as Sessions, Sessions as Effects}. In
  \bibinfo{booktitle}{\emph{Proceedings of the 43rd Annual ACM SIGPLAN-SIGACT
  Symposium on Principles of Programming Languages ( POPL'16)}}.
  \bibinfo{publisher}{Association for Computing Machinery},
  \bibinfo{pages}{568--581}.
\newblock
\urldef\tempurl%
\url{https://doi.org/10.1145/2837614.2837634}
\showDOI{\tempurl}


\bibitem[\protect\citeauthoryear{Perrin and Pin}{Perrin and Pin}{1995}]%
        {PP:semigroups:automata}
\bibfield{author}{\bibinfo{person}{Dominique Perrin} {and}
  \bibinfo{person}{Jean-Eric Pin}.} \bibinfo{year}{1995}\natexlab{}.
\newblock \showarticletitle{Semigroups and automata on infinite words}.
\newblock In \bibinfo{booktitle}{\emph{Semigroups, Formal Languages and
  Groups}}, \bibfield{editor}{\bibinfo{person}{John Fountain}} (Ed.).
  \bibinfo{series}{NATO ASI Series C: Mathematical and Physical Sciences},
  Vol.~\bibinfo{volume}{466}. \bibinfo{publisher}{Kluwer Academic Publishers},
  \bibinfo{pages}{49--72}.
\newblock


\bibitem[\protect\citeauthoryear{Sable~Group}{Sable~Group}{[n.d.]}]%
        {Soot}
\bibfield{author}{\bibinfo{person}{McGill~University Sable~Group}.}
  \bibinfo{year}{[n.d.]}\natexlab{}.
\newblock \bibinfo{title}{{Soot - {A} framework for analyzing and transforming
  {J}ava and {A}ndroid applications}}.
\newblock
\newblock
\urldef\tempurl%
\url{https://soot-oss.github.io/soot/}
\showURL{%
Retrieved July 22, 2021 from \tempurl}


\bibitem[\protect\citeauthoryear{Seidl and Vojdani}{Seidl and Vojdani}{2009}]%
        {10.1007/978-3-642-03237-0_13}
\bibfield{author}{\bibinfo{person}{Helmut Seidl} {and} \bibinfo{person}{Vesal
  Vojdani}.} \bibinfo{year}{2009}\natexlab{}.
\newblock \showarticletitle{Region Analysis for Race Detection}. In
  \bibinfo{booktitle}{\emph{Proceedings of the 16th International Symposium on
  Static Analysis}} \emph{(\bibinfo{series}{SAS '09})}.
  \bibinfo{publisher}{Springer-Verlag}, \bibinfo{address}{Berlin, Heidelberg},
  \bibinfo{pages}{171--187}.
\newblock
\urldef\tempurl%
\url{https://doi.org/10.1007/978-3-642-03237-0_13}
\showDOI{\tempurl}


\bibitem[\protect\citeauthoryear{Skalka}{Skalka}{2008}]%
        {10.1007/s10990-008-9032-6}
\bibfield{author}{\bibinfo{person}{Christian Skalka}.}
  \bibinfo{year}{2008}\natexlab{}.
\newblock \showarticletitle{Types and trace effects for object orientation}.
\newblock \bibinfo{journal}{\emph{Higher-Order and Symbolic Computation}}
  \bibinfo{volume}{21} (\bibinfo{year}{2008}), \bibinfo{pages}{239--282}.
\newblock
\urldef\tempurl%
\url{https://doi.org/10.1007/s10990-008-9032-6}
\showDOI{\tempurl}


\bibitem[\protect\citeauthoryear{Skalka, Smith, and Van~Horn}{Skalka
  et~al\mbox{.}}{2008}]%
        {10.1017/S0956796807006466}
\bibfield{author}{\bibinfo{person}{Christian Skalka}, \bibinfo{person}{Scott
  Smith}, {and} \bibinfo{person}{David Van~Horn}.}
  \bibinfo{year}{2008}\natexlab{}.
\newblock \showarticletitle{Types and Trace Effects of Higher Order Programs}.
\newblock \bibinfo{journal}{\emph{Journal of Functional Programming}}
  \bibinfo{volume}{18}, \bibinfo{number}{2} (\bibinfo{year}{2008}),
  \bibinfo{pages}{179--249}.
\newblock
\urldef\tempurl%
\url{https://doi.org/10.1017/S0956796807006466}
\showDOI{\tempurl}


\bibitem[\protect\citeauthoryear{Suzuki, Fujima, Kobayashi, and Tsukada}{Suzuki
  et~al\mbox{.}}{2017}]%
        {suzuki_et_al:LIPIcs:2017:7732}
\bibfield{author}{\bibinfo{person}{Ryota Suzuki}, \bibinfo{person}{Koichi
  Fujima}, \bibinfo{person}{Naoki Kobayashi}, {and} \bibinfo{person}{Takeshi
  Tsukada}.} \bibinfo{year}{2017}\natexlab{}.
\newblock \showarticletitle{Streett Automata Model Checking of Higher-Order
  Recursion Schemes}. In \bibinfo{booktitle}{\emph{2nd International Conference
  on Formal Structures for Computation and Deduction (FSCD'2017)}}
  \emph{(\bibinfo{series}{Leibniz International Proceedings in Informatics
  (LIPIcs)})}, \bibfield{editor}{\bibinfo{person}{Dale Miller}} (Ed.),
  Vol.~\bibinfo{volume}{84}. \bibinfo{publisher}{Schloss
  Dagstuhl--Leibniz-Zentrum fuer Informatik}, \bibinfo{pages}{32:1--32:18}.
\newblock
\urldef\tempurl%
\url{https://doi.org/10.4230/LIPIcs.FSCD.2017.32}
\showDOI{\tempurl}


\bibitem[\protect\citeauthoryear{Tofte and Talpin}{Tofte and Talpin}{1997}]%
        {10.1006/inco.1996.2613}
\bibfield{author}{\bibinfo{person}{Mads Tofte} {and}
  \bibinfo{person}{Jean-Pierre Talpin}.} \bibinfo{year}{1997}\natexlab{}.
\newblock \showarticletitle{Region-Based Memory Management}.
\newblock \bibinfo{journal}{\emph{Information and Computation}}
  \bibinfo{volume}{132}, \bibinfo{number}{2} (\bibinfo{year}{1997}),
  \bibinfo{pages}{109--176}.
\newblock
\urldef\tempurl%
\url{https://doi.org/10.1006/inco.1996.2613}
\showDOI{\tempurl}


\bibitem[\protect\citeauthoryear{Ye, Sui, and Xue}{Ye et~al\mbox{.}}{2014}]%
        {10.1007/978-3-319-10936-7_20}
\bibfield{author}{\bibinfo{person}{Sen Ye}, \bibinfo{person}{Yulei Sui}, {and}
  \bibinfo{person}{Jingling Xue}.} \bibinfo{year}{2014}\natexlab{}.
\newblock \showarticletitle{Region-Based Selective Flow-Sensitive Pointer
  Analysis}. In \bibinfo{booktitle}{\emph{Static Analysis (SAS 2014)}}
  \emph{(\bibinfo{series}{Lecture Notes in Computer Science})},
  \bibfield{editor}{\bibinfo{person}{M.~Müller-Olm} {and}
  \bibinfo{person}{H.~Seidl}} (Eds.), Vol.~\bibinfo{volume}{8723}.
  \bibinfo{publisher}{Springer}, \bibinfo{address}{Cham},
  \bibinfo{pages}{319--336}.
\newblock
\urldef\tempurl%
\url{https://doi.org/10.1007/978-3-319-10936-7_20}
\showDOI{\tempurl}


\bibitem[\protect\citeauthoryear{Zeng and Tan}{Zeng and Tan}{2016}]%
        {ZT:test:buechi}
\bibfield{author}{\bibinfo{person}{Bolong Zeng} {and} \bibinfo{person}{Li
  Tan}.} \bibinfo{year}{2016}\natexlab{}.
\newblock \showarticletitle{Test Reactive Systems with
  {B}\"uchi-Automaton-Based Temporal Requirements}. In
  \bibinfo{booktitle}{\emph{Theoretical Information Reuse and Integration}}
  \emph{(\bibinfo{series}{Advances in Intelligent Systems and Computing})},
  \bibfield{editor}{\bibinfo{person}{T.~Bouabana-Tebibel} {and}
  \bibinfo{person}{S.~Rubin}} (Eds.), Vol.~\bibinfo{volume}{446}.
  \bibinfo{publisher}{Springer, Cham}, \bibinfo{pages}{31--57}.
\newblock
\urldef\tempurl%
\url{https://doi.org/10.1007/978-3-319-31311-5_2}
\showDOI{\tempurl}


\end{thebibliography}

\appendix

%\clearpage

\section{Operational Semantics Rules}
\label{sec:operational}
We define an auxiliary function $\classOf_h(v)$ to determine the type of value $v$ in heap $h$ as follows:
\[
\classOf_h(v) :=
\begin{cases}
\NullType & \text{if } v=\itnull \\
C & \text{if } h(v)=(C,\_,\_)\in\Obj
\end{cases}
\]
The operational semantics $(s,h) \vdash e \Downarrow v,h' \andrel w$ of terminating evaluation is defined inductively from the following inference rules:
\[
\inferrule*
{ }
{(s,h) \vdash \tto(a) \Downarrow \itnull,h \andrel a}
\]
\[
\inferrule*
{ }
{(s,h) \vdash x \Downarrow s(x),h \andrel \varepsilon}
\qquad
\inferrule*
{ }
{(s,h) \vdash \ttnull \Downarrow \itnull, h \andrel \varepsilon}
\]
\[
\inferrule*
{s(x)=s(y)\\
 (s,h) \vdash e_1 \Downarrow v,h' \andrel w}
{(s,h) \vdash \ttifthenelse{x}{y}{e_1}{e_2} \Downarrow v,h' \andrel w}
\]
\[
\inferrule*
{s(x)\not=s(y)\\
 (s,h) \vdash e_2 \Downarrow v,h' \andrel w}
{(s,h) \vdash \ttifthenelse{x}{y}{e_1}{e_2} \Downarrow v,h' \andrel w}
\]
\[
\inferrule*
{(s,h) \vdash e_1 \Downarrow v_1 , h_1 \andrel w_1\\
 (s[x\mapsto v_1],h_1) \vdash e_2 \Downarrow v_2 , h_2 \andrel w_2}
{(s,h) \vdash \letin{x}{e_1}{e_2} \Downarrow v_2 , h_2 \andrel w_1 w_2}
\]
\[
\inferrule*
{l \not\in \dom(h)\\
 G=[f \mapsto \itnull]_{f\in\fields(C)}}
{(s,h) \vdash \new{\ell}{C} \Downarrow l, h[l\mapsto (C,G,\ell)] \andrel \varepsilon}
\]
\[
\inferrule*
{(s,h) \vdash e \Downarrow v,h' \andrel w\\
\classOf_{h'}(v) \preceq C}
{(s,h) \vdash \cast{C}{e} \Downarrow v, h' \andrel w}
\]
\[
\inferrule*
{s(x)=l\\
 h(l)=(\_,G,\_)}
{(s,h) \vdash x.f \Downarrow G(f), h \andrel \varepsilon}
\]
\[
\inferrule*
{s(x)=l\\
 h(l)=(D,G,\ell)\\
 h'=h[l \mapsto (D, G[f \mapsto s(y)], \ell)]}
{(s,h) \vdash x.f := y \Downarrow s(y), h' \andrel \varepsilon}
\]
\[
\inferrule*
{s(x)=l\qquad
 h(l) = (D,\_,\_)\qquad
 \mtable(D,m) = (\bar{z},e)\\
 ([\this \mapsto l]\cup[z_i \mapsto s(y_i)]_{i \in \{1,\ldots,|\bar{z}|\}},h) \vdash e \Downarrow v,h' \andrel w}
{(s,h) \vdash x.m(\bar{y}) \Downarrow v, h' \andrel w}
\]
The operational semantics of divergence $(s,h) \vdash e \Uparrow \andrel w$ is defined coinductively from the following inference rules, where double horizontal lines are used:
\[
\mprset{fraction={===}}
\inferrule*
{s(x)=s(y)\\
 (s,h) \vdash e_1 \Uparrow \andrel w}
{(s,h) \vdash \ttifthenelse{x}{y}{e_1}{e_2} \Uparrow \andrel w}
\]
\[
\mprset{fraction={===}}
\inferrule*
{s(x)\not=s(y)\\
 (s,h) \vdash e_2 \Uparrow \andrel w}
{(s,h) \vdash \ttifthenelse{x}{y}{e_1}{e_2} \Uparrow \andrel w}
\]
\[
\mprset{fraction={===}}
\inferrule*
{(s,h) \vdash e_1 \Downarrow v_1 , h_1 \andrel w_1\\
 (s[x\mapsto v_1],h_1) \vdash e_2 \Uparrow \andrel w_2}
{(s,h) \vdash \letin{x}{e_1}{e_2} \Uparrow \andrel w_1 w_2}
\]
\[
\mprset{fraction={===}}
\inferrule*
{(s,h) \vdash e_1 \Uparrow \andrel w}
{(s,h) \vdash \letin{x}{e_1}{e_2} \Uparrow \andrel w}
\]
\[
\mprset{fraction={===}}
\inferrule*
{(s,h) \vdash e \Uparrow \andrel w}
{(s,h) \vdash \cast{C}{e} \Uparrow \andrel w}
\]
\[
\mprset{fraction={===}}
\inferrule*
{s(x)=l \qquad
 h(l) = (D,\_,\_) \qquad
 \mtable(D,m) = (\bar{z},e)\\
 ([\this \mapsto l]\cup[z_i \mapsto s(y_i)]_{i \in \{1,\ldots,|\bar{z}|\}},h) \vdash e \Uparrow \andrel w}
{(s,h) \vdash x.m(\bar{y}) \Uparrow \andrel w}
\]

\section{Exception Handling}
\label{sec:exception}

For exception handling, we extend the syntax of FJ (Section~\ref{sec:FJ}) with expressions
\[
\throw e \quad \text{and} \quad \trycatch{e_1}{C\, x}{e_2}.
\]
The operational semantics (Section~\ref{sec:operational}) is also extended with a relation $(s,h)\vdash e \downharpoonright l,h' \andrel w$ to model exception throwing. It expresses that, in the state $(s,h)$, the evaluation of the expression $e$ throws an exception that is stored in the location $l$, with the heap updated to $h'$, generating the event trace $w \in \Sigma^*$. The operational semantics of terminating evaluation $\Downarrow$ and the one of exception throwing $\downharpoonright$ are defined mutually. In particular, we have the following rules for the $\mathtt{throw}$ and $\mathtt{try}$-$\mathtt{catch}$ expressions:
\[
\inferrule*
{(s,h) \vdash e \Downarrow l,h' \andrel w}
{(s,h) \vdash \throw e \downharpoonright l, h' \andrel w}
\qquad
\inferrule*
{(s,h) \vdash e \downharpoonright l,h' \andrel w}
{(s,h) \vdash \throw e \downharpoonright l, h' \andrel w}
\]
\[
\inferrule*
{(s,h) \vdash e_1 \Downarrow v,h' \andrel w}
{(s,h) \vdash \trycatch{e_1}{E\,x}{e_2} \Downarrow v, h' \andrel w}
\]
\[
\inferrule*
{(s,h) \vdash e_1 \downharpoonright l, h_1 \andrel w_1\\
 \classOf_{h_1}(l) \preceq E\\
 (s[x \mapsto l],h_1) \vdash e_2 \Downarrow v,h_2 \andrel w_2}
{(s,h) \vdash \trycatch{e_1}{E\,x}{e_2} \Downarrow v,h_2 \andrel w_1w_2}
\]
\[
\inferrule*
{(s,h) \vdash e_1 \downharpoonright l, h' \andrel w\\
 \classOf_{h'}(l) \not\preceq E}
{(s,h) \vdash \trycatch{e_1}{E\,x}{e_2} \downharpoonright l, h' \andrel w}
\]
\[
\inferrule*
{(s,h) \vdash e_1 \downharpoonright l_1, h_1 \andrel w_1\\
 \classOf_{h_1}(l_1) \preceq E\\
 (s[x \mapsto l_1],h_1) \vdash e_2 \downharpoonright l_2, h_2 \andrel w_2}
{(s,h) \vdash \trycatch{e_1}{E\,x}{e_2} \downharpoonright l_2, h_2 \andrel w_1w_2}
\]
\[
\mprset{fraction={===}}
\inferrule*
{(s,h) \vdash e \Uparrow \ \andrel w}
{(s,h) \vdash \throw{e} \Uparrow \ \andrel w}
\]
\[
\mprset{fraction={===}}
\inferrule*
{(s,h) \vdash e_1 \Uparrow \ \andrel w}
{(s,h) \vdash \trycatch{e_1}{E\,x}{e_2} \Uparrow \ \andrel w}
\]
\[
\mprset{fraction={===}}
\inferrule*
{(s,h) \vdash e_1 \downharpoonright l, h_1 \andrel w_1\\
 \classOf_{h_1}(l) \preceq E\\\\
 (s[x \mapsto l],h_1) \vdash e_2 \Uparrow \ \andrel w_2}
{(s,h) \vdash \trycatch{e_1}{E\,x}{e_2} \Uparrow \ \andrel w_1w_2}
\]
The other rules for \eg~$\mathtt{if}$- and $\mathtt{let}$-statements are standard.

Instead of a complete redevelopment, we explain only the main changes to the type system introduced in Section~\ref{sec:type:system:sets}.

Firstly, to model also exception handling, typing judgments have the following form:
\[
\Gamma \vdash e : T \throws H \calls S
\]
where the expression $T$ for terminating executions and the call expression $S$ have been explained in Section~\ref{sec:formal:expression}. The expression $H \in \cP(\Sigma^*)\langle\Reg\rangle$ encodes the possible exceptional behaviors of $e$. If $H = r_1 \andrel U_1 \mid \cdots \mid r_n \andrel U_n$, then it expresses that, whenever $e$ throws an exception, there is some $i \in \{ 1, \ldots, n \}$ such that the value of the exception is in region $r_i$ and the generated event trace is in $U_i$. For instance, if $e$ is the follow expression
\[
\begin{aligned}
& \mathtt{let}\ x = \ttifthenelseg{\mathit{cond}}{(\tto(a);\, \new{\ell_1}{C})}{(\new{\ell_2}{D})}\\
& \mathtt{in}\ \tto(b);\, \throw x
\end{aligned}
\]
then both $T$ and $S$ can be $\emptyset$ because $e$ doesn't terminate properly and has no method calls, and $H$ can be $\Createdat{\ell_1} \andrel \{ab\} \mid \Createdat{\ell_2} \andrel \{b\}$ meaning that either $e$ throws an exception that was created with label $\ell_1$ and generates trace $ab$, or it throws an exception that was created with label $\ell_2$ and generates trace $b$.

The method typing~$M$ encodes also the expressions for exceptional behaviors of the methods; thus, it assigns to each method signature a tuple $(T,H,S)$ of formal expressions. The typing rules in Figure~\ref{fig:typing} are adapted according to the possible exceptional behaviors of the expressions. The base cases always have the empty expression of exceptional behaviors such as in
\[
\inferrule*[left={new}]
{\ }
{\Gamma \vdash \new{\ell}{C} : \Createdat{\ell} \andrel \{\varepsilon\} \throws \emptyset \calls \emptyset}
\]
because no exceptions can be thrown. Here are the typing rules for the $\mathtt{throw}$ and $\mathtt{try}$-$\mathtt{catch}$ expressions:
\[
\inferrule*[left={throw}]
{\Gamma \vdash e : T \throws H \calls S}
{\Gamma \vdash \throw{e} : \emptyset \throws T \cup H \calls S}
\]
\[
\inferrule*[left={trycatch}]
{ \Gamma \vdash e_1 : T \throws H \calls S \\
 H = r_1 \andrel U_1 \mid \ldots \mid r_n \andrel U_n \\\\
 \text{for all } i \in \{ 1, \ldots, n \} \ (\\\\
 C \preceq E \text{ for some } C \in \Cls(r_i) \text{\ \ and} \\
 \Gamma,\, x\of r_i \vdash e_2 : T_i \throws H_i \calls S_i \\\\
 \text{or}\\\\
 C \not\preceq E \text{ for all } C \in \Cls(r_i) \text{\ \ and \ }
 T_i = H_i = S_i = \emptyset)\\\\
 T' = T \cup \textstyle\bigcup^n_{i=1} U_i \cdot T_i \\\\
 H' = H/E \cup \textstyle\bigcup^n_{i=1} U_i \cdot H_i \\\\
 S' = S \cup \textstyle\bigcup^n_{i=1} U_i \cdot S_i}
{\Gamma \vdash \trycatch{e_1}{E\,x}{e_2} : T' \throws H' \calls S'}
\]
Rule \textsc{trycatch} is similar to rule \textsc{let}, but $e_2$ is typed only with the regions $r_i$ of $e_1$ to which the exception class $E$ is related, \ie, $E$ or its subclass is allowed by $r_i$. Moreover, because all the exceptions of class $E$ or its subclass that $e_1$ may throw would be caught, we should remove the related exception effects in $H$ before joining them into the resulting one $H'$. For this, we define an operation $H/E$ to remove from $H$ those pairs $(r \andrel U)$ such that all the classes allowed by $r$ are below $E$, \ie, we define
\[
H/E = H/\{ (r \andrel U) \in H \mid C \preceq E \text{ for all } C \in \Cls(r) \}
\]
where $H$ is treated as a set of pairs $(r \andrel U)$ and thus $H/E$ is defined using set complement. For example, suppose $\Cls(r)=\{C,D\}$ and $\Cls(s)=\{C,E\}$ with $C \preceq E$ and $D \not\preceq E$, then $(r \andrel U \mid s \andrel V) / E = (r \andrel U)$ because objects in region $s$ must be of class $E$ or its subclass.

Lastly, the soundness statement (Theorem~\ref{thm:soundness:gen}) states that the typing also correctly captures the exceptional behaviors, \ie, if
\[
\Gamma \vdash e : T \throws H \calls S
\quad \text{and} \quad
(s,h) \vdash e \downharpoonright l,h' \andrel w
\]
then there is $(r \andrel U) \in H$ such that $(l,h') \vdash r$ and $w \in U$. 

\end{document}